\documentclass[letterpaper, 10pt, conference]{ieeeconf}      % Use this line for a4
\IEEEoverridecommandlockouts                              % This command is only
\usepackage{amsmath}

\usepackage{amsthm}
\usepackage{amsfonts}
\usepackage{xcolor}
\makeatletter
\def\endfigure{\end@float}
\def\endtable{\end@float}
\makeatother
\usepackage{cite}
\usepackage[caption=false,font=footnotesize]{subfig}
\usepackage{graphicx}

\usepackage[hidelinks]{hyperref} % (swap your current hyperref line to this)
\makeatletter
\usepackage{fix-cm}
\makeatother
\newtheorem{lemma}{Lemma}
\newtheorem{example}{Example}
\newtheorem{theorem}{Theorem}
\newtheorem{remark}{Remark}
\newtheorem{corollary}{Corollary}
\newtheorem{assumption}{Assumption}
\newtheorem{problem}{Problem}
\newtheorem{lemmma}{Lemma}
\newcommand{\R}{\mathbb{R}}     %real numbers
\newcommand{\beq}[1]{\begin{align*}\label{eq:#1}}
\newcommand{\eeq}{\end{align*}}

\makeatletter
\newcommand{\xMapsto}[2][]{\ext@arrow 0599{\Mapstofill@}{#1}{#2}}
\def\Mapstofill@{\arrowfill@{\Mapstochar\Relbar}\Relbar\Rightarrow}
\makeatother

\title{\LARGE \bf 
Beyond Quadratic Costs: A Bregman Divergence Approach to \texorpdfstring{$H_\infty$}{H∞} Control
}

\author{Joudi Hajar, Reza Ghane, Babak Hassibi% <-this % stops a space
\thanks{The authors are in the Department of Electrical Engineering, Caltech. Emails:
        {\tt\small \{jhajar,rghanekh,hassibi\}@caltech.edu}}%
}

\usepackage{comment}
\usepackage{etoolbox}

\AtBeginEnvironment{align}{%
  \setlength{\abovedisplayskip}{1.5pt}%
  \setlength{\abovedisplayshortskip}{1.5pt}%
}

\AtEndEnvironment{align}{%
  \setlength{\belowdisplayskip}{1.5pt}%
  \setlength{\belowdisplayshortskip}{1.5pt}%
}

\AtBeginEnvironment{equation}{%
  \setlength{\abovedisplayskip}{1.5pt}%
  \setlength{\abovedisplayshortskip}{1.5pt}%
}

\AtEndEnvironment{equation}{%
  \setlength{\belowdisplayskip}{1.5pt}%
  \setlength{\belowdisplayshortskip}{1.5pt}%
}

\AtBeginEnvironment{aligned}{%
  \setlength{\abovedisplayskip}{1.5pt}%
  \setlength{\abovedisplayshortskip}{1.5pt}%
}

\AtEndEnvironment{aligned}{%
  \setlength{\belowdisplayskip}{1.5pt}%
  \setlength{\belowdisplayshortskip}{1.5pt}%
}
\begin{document}

\maketitle
\thispagestyle{empty}
\pagestyle{empty}

%%%%%%%%%%%%%%%%%%%%%%%%%%%%%%%%%%%%%%%%%%%%%%%%%%%%%%%%%%%%%%%%%%%%

\begin{abstract}

In the past couple of decades, non-quadratic convex penalties have reshaped signal processing and machine learning; in robust control, however, general convex costs break the Riccati and storage function structure that make the design tractable. Practitioners thus default to approximations, heuristics or robust model predictive control that are solved online for short horizons. We close this gap by extending $H_\infty$ control of discrete-time linear systems to strictly convex penalties on state, input, and disturbance, recasting the objective with Bregman divergences that admit a completion-of-squares decomposition. The result is a closed-form, time-invariant, full-information stabilizing controller that minimizes a worst-case performance ratio over the infinite horizon. Necessary and sufficient existence/optimality conditions are given by a Riccati-like identity together with a concavity requirement; with quadratic costs, these collapse to the classical $H_\infty$ algebraic Riccati equation and the associated negative-semidefinite condition, recovering  the linear central controller. Otherwise, the optimal controller is nonlinear and can enable safety envelopes, sparse actuation, and bang–bang policies with rigorous $H_\infty$ guarantees. 
\end{abstract}

%%%%%%%%%%%%%%%%%%%%%%%%%%%%%%%%%%%%%%%%%%%%%%%%%%%%%%%%%%%%%%%%%%%%%%%%%%%%%%%%
\section{Introduction and Related Works}
Robust control has evolved markedly over recent decades, from foundational ideas to mature state-space formulations and beyond. Beginning with the landmark work of \cite{zames_feedback_1981}, the \(H_\infty\) paradigm was refined for linear state-space models in the 1980s \cite{zames2003feedback,kimura1987directional,doyle_state-space_1988,zhou1996robust} and extended to nonlinear models in the 1990s \cite{van19922}. It is now a cornerstone of robust control with applications in robotics \cite{gadewadikar2009h}, electric machines \cite{rigatos2015nonlinear}, and traffic control \cite{zhang2022robust}. Classically, $H_\infty$ control of linear systems is posed with quadratic penalties on state, input, and disturbance, and the controller is obtained by minimizing a worst-case ratio that bounds the system’s \(L_2\) gain. Under this quadratic structure, algebraic Riccati equations (AREs) and completion-of-squares identities yield \emph{closed-form}, scalable solutions, with feasibility certified by a negative-semidefinite condition \cite{blackbook}. Recently, quadratic tractability has also enabled a \emph{closed-form} frequency-domain characterization of distributionally robust controllers that interpolate between $H_2$ and $H_\infty$ \cite{karginl4dc,karginICML,hajarcdc,hajar_wasserstein_2023}. Yet the quadratic template can be overly rigid: it cannot encode side information about disturbance geometry (e.g., sparsity or $\ell_\infty$ bounds) nor promote behaviors like sparse or saturating actuation without ad hoc workarounds. In practice, when richer structure matters, practitioners often turn to robust MPC—min–max, tube, or LMI formulations solved online over short horizons \cite{ScokaertMayne1998,RawlingsMayneDiehl2017,KouvaritakisCannon2015}.

Meanwhile, adjacent fields like signal processing and machine learning routinely exploit nonquadratic convex penalties (e.g., \(\ell_1\) regularization, cross-entropy), whereas explicit closed-form synthesis in control has largely remained tied to quadratics (with a few dynamic-programming strands providing explicit solutions in specialized settings \cite{mare2005analytical,jones2021generalization,rantzer2022explicit}). Our goal here is to bring comparable flexibility to \emph{infinite-horizon} robust $H_\infty$ control in the discrete-time, full-information, LTI setting, without sacrificing classical closed-form solutions. Complementing our companion paper \cite{ourpaper2} that generalized the LQR framework to nonquadratic costs, we extend the quadratic \(H_\infty\) design to strictly convex costs and derive the central controller by reformulating the objective using \emph{Bregman divergences} \cite{bregman1967relaxation}, which quantify the gap between a convex function and its first-order approximation. The key step mirrors the classical derivation \cite{blackbook}: quadratic completion-of-squares is replaced by the Bregman three-point identity—the nonquadratic analogue of completion-of-squares— yielding a family of $H_\infty$ controllers that achieve a performance level $\gamma>0$ (the upper bound on the system's gain).

%%%%%%%%%%%%%%%%%%%%%%%%%
The primary contributions of this paper are as follows:
\begin{itemize}
\item \textbf{Nonquadratic \(H_\infty\) cost:} We extend the classical quadratic framework to strictly convex functions of the state, input, and disturbance enabling more sophisticated and desirable controller behaviors.
    \item \textbf{Closed-form stabilizing controller:} By reformulating the cost using \emph{Bregman divergences}, we derive a time-invariant, closed-form, stabilizing central controller that achieves a performance level $\gamma$ over the infinite horizon, and provide a complete parametrization of all controllers achieving the desired performance.
    \item \textbf{Existence and optimality:} We provide \emph{necessary and sufficient} conditions for the existence of a family of robust controllers via a concavity condition and Riccati-like equations. 
    %\item \textbf{Stability:} We prove that the resulting controller is stabilizing using Lyapunov stability. %(bounded-energy disturbance) 
    \item \textbf{Classical $H_\infty$ Recovery:} Under quadratic costs, our conditions reduce exactly to the standard ARE and the classical linear controller.
    \item \textbf{Computational feasibility \& recipe:} We give an offline synthesis recipe; a single feasibility program yields a time-invariant controller (no online optimization).
    \item \textbf{Illustrative behaviors:} We point to safety envelopes, sparse actuation, and saturating policies as natural outcomes of nonquadratic costs; due to space, we include a single case study (input-limited control) that saturates inputs by design with a rigorous $H_\infty$ guarantee.

\end{itemize}

\paragraph*{Links to dynamic-game formulations} Game-theoretic approaches pose an auxiliary min–max dynamic game whose saddle point recovers the familiar central controller for the quadratic $H_\infty$ problem \cite{bacsar2008h}. However, the dynamic game is contrived relative to the classical $H_\infty$ formulation: it restricts the disturbance to be \emph{causal} and optimizes a surrogate objective, whereas classical $H_\infty$ certifies robustness against \emph{all} disturbances, independent of causality. Therefore, we extend the \emph{classical} (noncausal) setting beyond quadratics, and deliberately depart from the game-theoretic framework (even though a similar min–max interpretation applies and its saddle point recovers the central controller in our nonquadratic extension). Using Bregman divergences, we derive not only sufficient conditions, but also necessary ones for the existence of a \emph{family} of $H_\infty$ controllers. We thus ensure robust performance in the true \(H_\infty\) sense.

In what follows, we develop the theory, derive the central controller, and discuss the implications of our generalized framework across various application domains.

\section{Preliminaries}
\subsection{Notations}
The set of real numbers is denoted by \(\mathbb{R}\). The \(\ell_2\)-norm is written as \(\|\cdot\|_2\). For any matrix \(X\), its transpose is denoted by \(X^\top\) and its inverse by $X^{-1}$. For symmetric matrices \(A\) and \(B\), we write \(A \succeq B\) if \(A - B\) is positive semidefinite, and \(A \succ B\) if \(A - B\) is positive definite. $I$ is the identity matrix of appropriate dimensions.

For a function \(f:\mathbb{R}^n \to \mathbb{R}\), its Fenchel dual is defined as 
\[
f^\ast(y) = \sup_{x\in\mathbb{R}^n} \{ x^\top y - f(x) \}.
\]
\vspace{-4mm}

Moreover, \(\nabla f(\cdot)\) denotes the gradient of \(f(\cdot)\), \(\nabla^2 f(\cdot)\) its Hessian, and \(\bigl(\nabla f\bigr)^{-1}(\cdot)\) the inverse of the gradient operator.
\subsection{Quadratic \texorpdfstring{$H_\infty$}{H∞} Framework and Established Results}
We consider the discrete-time linear system described by
\begin{equation}\label{eq:statespace}
    x_{k+1} = A x_k + B u_k + w_k,
\end{equation}
where \(x_k \in \mathbb{R}^n\) denotes the state, \(u_k \in \mathbb{R}^d\) is the control input, and \(w_k \in \mathbb{R}^n\) represents the disturbance. For simplicity, $x_0=0$. %The disturbance sequence has bounded energy, i.e.,
%$\|w_k\|_2^2 < \infty$ for every $k\ge1$, and .

As background, we briefly recall the finite-horizon quadratic 
$H_\infty$ formulation, which we use as a reference point for our nonquadratic, infinite-horizon results. In quadratic 
$H_\infty$ control, the objective is to design a sequence of causal control laws that guarantees robust performance against any disturbance sequence \(\{w_k\}_{k=0}^{T}\). This leads to the following optimization problem:
\begin{equation}\label{eq:optimization}
    \begin{aligned}
         &\quad\quad\quad\quad\quad\quad\quad\quad \min_{\substack{\text{causal} \{u_k\}}} \; \max_{\substack{\{w_k\}\neq 0  }} \;J_Q,\\
    &J_Q:=\frac{x_{T+1}^\top P_{T+1} x_{T+1} + \sum\limits_{k=0}^{T} \left( x_k^\top Q_k x_k + u_k^\top R_k u_k \right)}{\sum\limits_{k=0}^{T} w_k^\top S_k w_k},
    \end{aligned}
\end{equation}
with $Q_k,R_k,S_k,P_{T+1}\succeq 0$. To handle complexity, the standard $\gamma$-relaxation is used to recast the problem as follows.
\begin{problem}\label{pb:1HI}
    For a given performance level \(\gamma > 0\), find a sequence of causal controllers \(\{u_k\}\) such that for every disturbance sequence \(\{w_k\}_{k=0}^{T}\), the following performance bound is achieved:
    \begin{equation}\label{eq:performance_bound}
        \frac{x_{T+1}^\top P_{T+1} x_{T+1} + \sum\limits_{k=0}^{T} \left( x_k^\top Q_k x_k + u_k^\top R_k u_k \right)}{\sum\limits_{k=0}^{T} w_k^\top S_k w_k} \leq \gamma^2.
    \end{equation}
\end{problem}

By employing completion-of-squares along with appropriate matrix factorizations, one can derive not only a nominal (or ``central'') controller but also a family of controllers that ensure the desired robust performance, as described by the following result from \cite{blackbook}.

\begin{theorem}\label{thm:robust_ctrl}
A full-information causal controller strategy $\check{u}_k = \mathcal{F}_k(w_0, w_1, \dots, w_k)$,
that solves Problem \ref{pb:1HI} and achieves the desired \(H_\infty\) performance exists if and only if, for each time step \(k\), the matrix
\begin{equation}\label{eq:Delta}
    \Delta_k := -\gamma^2 S_k + P_{k+1}^c - P_{k+1}^c B\, R_{B^c,k}^{-1} B^\top P_{k+1}^c
\end{equation}
is negative semidefinite, i.e., $\Delta_k\preceq 0$, where $R_{B^c,k} := R_k + B^\top P_{k+1}^c B,$ and the sequence \(\{P_k^c\}\) is obtained by the backward Riccati recursion
\begin{equation}\label{eq:riccati}
    P_k^c = A^\top P_{k+1}^c A + Q_k - K_{c,k}^\top R_{e,k}^c K_{c,k},
\end{equation}
with the feedback gain defined by
%\begin{equation}\label{eq:Kck}
$K_{c,k} = \bigl(R_{e,k}^c\bigr)^{-1} \begin{pmatrix} B^\top \\ I \end{pmatrix} P_{k+1}^c A,$
%\end{equation}
and $R_{e,k}^c = \begin{pmatrix}
        R_k + B^\top P_{k+1}^c B & B^\top P_{k+1}^c \\
        P_{k+1}^c B & -\gamma^2 S_k + P_{k+1}^c
    \end{pmatrix}.$

    In this setting, all possible full-information causal controllers $\check{u}_k= \mathcal{F}_k(w_0,\cdots,w_k)$ are given by those that satisfy
    \begin{align}\label{eq:compSQ}
        &\sum\limits_{k=0}^T\! (\check{u}_k \!-\!\bar{u}_k)^\top \!R_{B^c,k} (\check{u}_k\!\! -\bar{u}_k)\!+ \!(w_k \!-\!\hat{w}_k)^\top\! \Delta_k (w_k \!-\!\hat{w}_k)\! <\! 0 %\nonumber\\&+ \sum\limits_{k=0}^T \!(w_k -\hat{w}_k)^\top\! \Delta_k (w_k -\hat{w}_k)\! <\! 0
    \end{align}
    \vspace{-2mm}
    
    where $\begin{pmatrix}
        \hat{u}_k \\ \hat{w}_k
    \end{pmatrix}=-K_{c,k}x_k,$ 
     and $\bar{u}_k = -R^{-1}_{B^c,k} B^\top P_{k+1}^c (Ax_k + w_k).$
    The central controller is given by $\check{u}_k = \bar{u}_k.$
\end{theorem}
Theorem~\ref{thm:robust_ctrl} links the existence of a robust \(H_\infty\) controller to the solution of a Riccati recursion. It provides necessary and sufficient conditions for a full-information controller to exist and provides a complete parametrization of all controllers achieving the desired performance. For further details, see~\cite{blackbook}.

\subsection{Bregman Divergence}

The Bregman divergence is a measure of deviation of function from its linear approximation. For a strictly convex and differentiable function $\phi: \mathbb{R}^n \rightarrow \mathbb{R}$, and any two points $x, y \in \mathbb{R}^n$ we define the Bregman divergence \cite{bregman1967relaxation} as:
\begin{equation}
    D_\phi(x, y) = \phi(x) - \phi(y) - \nabla \phi(y)^\top (x - y).
    \label{eq:bregman_divergence}
\end{equation}
This divergence quantifies the difference between the value of the convex function at $x$ and its first-order Taylor approximation around $y$ evaluated at $x$. We present a popular example of the bregman divergence.
 \begin{example}\label{example:1}
    \emph{Squared Euclidean Distance:} Let $\phi(x) = \frac{1}{2} \|x\|_2^2$. Substituting into the definition of Bregman divergence yields:
$D_\phi(x, y) = \frac{1}{2} \|x\|_2^2 - \frac{1}{2} \|y\|_2^2 - y^\top (x - y) = \frac{1}{2} \|x - y\|_2^2.$
\end{example}
Bregman divergence possesses several key properties that facilitate its application in optimization and control theory. The following lemma summarizes these properties, where the "completion-of-squares"and "Duality" properties lie at the heart of the proofs in this work. 

\begin{lemma}\label{lemm:Breg}
    The following holds \cite{bregman1967relaxation,bregman2ML}:
    \begin{itemize}
        \item \emph{\textbf{Non-Negativity}}: $D_\phi(x, y) \geq 0$ for all $x,y \in \R^n$, with $D_\phi(x, y) = 0\Leftrightarrow$ $x = y$.
        \item \emph{\textbf{Asymmetry}}: In general, $D_\phi(x, y) \neq D_\phi(y, x)$.
        \item \emph{\textbf{Convexity}}: For a given $y$, the function $D_\phi(\cdot, y)$ is convex in its first entry.
        %\item \emph{\textbf{Law of Cosines:}}\textcolor{red}{remove} For any three points $x, y, z \in \mathbb{R}^n$, we have:      
        %D_\phi(x, y) = D_\phi(x, z) + D_\phi(z, y) - \left( \nabla \phi(y) - \nabla \phi(z) \right)^\top (x - z)
       % \begin{align*}%\label{eq:law_of_cosines}
        %   \!\!\!\!\! D_\phi(x, \!y) \!=\! D_\phi(x, \!z)\! +\! D_\phi(z,\! y)\! -\! \left( \nabla\! \phi(y) \!-\! \nabla\! \phi(z) \right)^\top (x\! -\! z).\!\!\!\!        
        %\end{align*}
        \item \emph{\textbf{Duality:}} $D_\phi(x,y) = D_{\phi^\ast } (\nabla\phi(y), \nabla\phi(x))$
        \item  \emph{\textbf{completion-of-squares (three-point identity):}} For two convex functions $\phi_1$ and $\phi_2$, and points $x, y, z \in \mathbb{R}^n$, we have the following:
        \begin{align}  %\label{eq:completion_of_squares}
            &D_{\phi_1}(x, y) + D_{\phi_2}(x, z) =\nonumber\\
             &D_{\phi_1 + \phi_2}\left(x, x^\ast\right)  + D_{\phi_1}\left(x^\ast, y\right) + D_{\phi_2}\left(x^\ast, z\right), 
        \end{align}
        where $x^\ast$ satisfies:$\nabla (\phi_1 + \phi_2)(x^\ast) = \nabla \phi_1(y) + \nabla \phi_2(z).$
        %\begin{equation}
         %   \nabla (\phi_1 + \phi_2)(x^\ast) = \nabla \phi_1(y) + \nabla \phi_2(z).
        %\end{equation}
    
    \end{itemize}
\end{lemma}

\section{\texorpdfstring{$H_\infty$}{H∞} Control in the Infinite Horizon with nonquadratic Costs}

We aim to move beyond conventional quadratic costs by allowing general strictly convex penalties on the state \(x_k\), control \(u_k\), and disturbance \(w_k\). Within this broader setting, we derive infinite-horizon closed-form controllers—akin to those in the quadratic case—that are applicable even with nonquadratic cost functions. This extension enables richer and more tailored control behaviors, addressing performance objectives that quadratic formulations cannot capture.

\subsection{Problem Setting and Assumptions}
We consider the discrete-time system in~\eqref{eq:statespace} and extend the infinite-horizon $H_\infty$ problem to allow general, nonquadratic penalties. The designer specifies the state cost $q(x_k)$, control cost $r(u_k)$, and disturbance cost $s(w_k)$, each assumed \textbf{strictly convex, even, and nonnegative}.

To capture the cost incurred after any finite horizon, we introduce a terminal cost (or storage function),
$
p(x).
$
Unlike the other cost functions, \(p(x)\) is \emph{not arbitrarily chosen} by the designer; instead, it is determined by the infinite-horizon optimality conditions, and it is assumed to be strictly convex and positive.  This property helps ensure both system stability and a well-posed performance measure. In the quadratic $H_\infty$ case, this reduces to $p(x) = x^\top P x$ with $P \succ 0$, the unique stabilizing solution of the $H_\infty$ ARE.

\begin{problem}\label{pb:bregman} 
We consider the following infinite-horizon robust control problem. Our goal is to design a causal time-invariant control policy $\{u_k\}$ where
$u_k = \mathcal{F}(w_0,w_1,\dots,w_k),$ which minimizes the worst-case \(H_\infty\) cost ratio over all disturbance sequences \(\{w_k\}\). That is,
\begin{align}\label{eq:breg_cost-Inf0}
&\min_{\text{causal } \{u_k\}} \lim_{T\rightarrow\infty}  \max_{\; \{w_k\}\neq 0}  J_{G},\\
&J_G:=\frac{\,p(x_{T+1}) + \sum\limits_{k=0}^{T} \bigl[q(x_k)+r(u_k)\bigr]}{\sum\limits_{k=0}^{T} s(w_k)}.
\end{align}
Here, \(q\), \(r\), and \(s\) are designer-specified cost functions that are strictly convex, positive, and even. The storage function \(p\) arises from infinite-horizon optimality conditions, and is also assumed to be strictly convex. 

\medskip
\textbf{Relaxed Problem:}

For a given performance level $\gamma>0$, find a causal time-invariant control policy \(\{u_k\}\) such that, for every finite horizon \(T\ge 1\),
\begin{equation}\label{eq:breg_cost-Inf}
    \!\!\!\!\max_{\;\{w_k\}\neq 0} \; \frac{\,p(x_{T+1}) + \sum\limits_{k=0}^{T} \bigl[q(x_k)+r(u_k)\bigr]}{\sum\limits_{k=0}^{T} s(w_k)} \le \gamma^2,
\end{equation}
which is equivalent to
$$\frac{\,p(x_{T+1}) \!+\! \sum\limits_{k=0}^{T} \bigl[q(x_k)\!+\!r(u_k)\bigr]}{\sum\limits_{k=0}^{T} s(w_k)} \le \!\gamma^2,\text{ } \forall \;\!\!\{w_k\}\neq 0.\!\!$$

\textbf{Final Statement:}
For a given performance level $\gamma>0$, find a causal time-invariant control policy \(\{u_k\}\) such that, for every finite horizon \(T\ge 1\), and for all admissible disturbance sequences \(\{w_k\}\).
\begin{equation}\label{eq:breg_cost-Inf2}
\!\!J_T := \sum\limits_{k=0}^{T}\left[\gamma^2 s(w_k)-q(x_k)-r(u_k)\right]-p(x_{T+1}) \ge 0 .%\quad \forall \{w_k\}.
\end{equation}
\end{problem}
This formulation extends the conventional quadratic-cost $H_\infty$ problem to general convex penalties, enabling synthesis of richer controller behaviors.

\subsection{Assumptions}
We make the following standing assumptions which simplify the analysis.
\begin{assumption}\label{asum:q}
$p,q,s:\mathbb{R}^n\!\to\!\mathbb{R}$ and $r:\mathbb{R}^d\!\to\!\mathbb{R}$ are strictly convex, even, and nonnegative, with
$p(0)=q(0)=r(0)=s(0)=0$ and $\nabla p(0)=\nabla q(0)=\nabla r(0)=\nabla s(0)=0$.
\end{assumption}
\begin{assumption}\label{asum:AB}
The system matrix \(A\) is invertible, and the input matrix $B$ is full-rank. 
\end{assumption}
Assumption \ref{asum:q} implies that  the fenchel duals of $p(\cdot),$ $q(\cdot),$ $r(\cdot)$ and $ s(\cdot)$, namely $p^\ast(\cdot),q^\ast(\cdot)$, $r^\ast(\cdot)$ and $ s^\ast(\cdot)$ inherit the same properties. Moreover, Assumption \ref{asum:AB} is mild since sampling a continuous-time system ensures $A$ is invertible, and any rank-deficient $B$ can be compressed to full-rank.

\section{Main results}
We now present our main theorem, which extends the quadratic $H_\infty$ framework to arbitrary strictly convex costs and provides \emph{necessary and sufficient conditions} for solvability. The result mirrors the structure of the classical framework and replaces the quadratic completion-of-squares step with its Bregman divergence analogue, leading to a Riccati-like identity, a disturbance–concavity condition, and a closed-form stabilizing central controller. In the quadratic case, these reduce exactly to the standard $H_\infty$ equations.

\subsection{Main Theorem}
\begin{theorem}\label{thm:main}
   Let Assumptions \ref{asum:q} and \ref{asum:AB} hold. Consider the state-space model in \eqref{eq:statespace}. Then, for any given $\gamma>0$, a causal time-invariant $H_\infty$ control policy $\{u_k\}$ that solves Problem \eqref{pb:bregman} exists if and only if:
\begin{equation}\label{eq:concave}
    -\gamma^2s(\cdot)+g(Ax+(\cdot)) \text{ is concave in }(\cdot), 
\end{equation}where 
\begin{equation}\label{eq:RICC1}
    g^\ast(\xi):= p^\ast(\xi)+r^\ast(B^\top\xi)
\end{equation}
and $p$ is a strictly convex Lyapunov function satisfying the Riccati-like equation:
\begin{equation}\label{eq:RICC2}
    p^\ast(\xi)+r^\ast(B^\top\xi)=(p-q)^\ast(A^\top\xi)+\gamma^2s^\ast(\gamma^{-2}\xi)
\end{equation}

In this case, all possible control strategies $\{u_k\}$, are given by those that satisfy,
\begin{align}\label{HELLO}
     &\sum\limits_{k=0}^{T} D_{\gamma^2s}(w_k,\hat w_k)\!-\!D_{g} \Bigl(Ax_k\!+\!w_k,Ax_k\!+\!\hat w_k\Bigr)\!-\!D_{r}(u_k,u_k^\ast)\!\nonumber\\&-D_{p}\Bigl(Ax_k+Bu_k+w_k,Ax_k+Bu_k^\ast+w_k\Bigr)\geq 0
\end{align}
where $\hat w_k$ is defined via 
\begin{align}
    \hat w_k=\nabla s^\ast\bigl(\gamma^{-2}A^{-\top}\nabla (p-q)(x_k)\bigr)
\end{align}
and $u^\ast_k$ via
\begin{align}
   & u^\ast_k= -\nabla r^\ast\big(B^\top\nabla g(Ax_k+w_k)\big).\label{eq:optimal_controller}
\end{align}

Finally, the so-called central controller is \emph{stabilizing}, and is given by $u_k=u_k^\ast$.
\end{theorem}

Theorem~\ref{thm:main} extends quadratic $H_\infty$ control (Theorem~\ref{thm:robust_ctrl}) to strictly convex costs while preserving the classical structure. The matrix negativity test \eqref{eq:Delta} is replaced by the concavity condition \eqref{eq:concave}, ensuring the worst-case disturbance is well posed and attained. The Riccati recursion \eqref{eq:riccati} is replaced by the dual identity \eqref{eq:RICC1}–\eqref{eq:RICC2}, enforcing the one-step storage balance in Fenchel-dual variables. The quadratic identity \eqref{eq:compSQ}, obtained from the completion-of-squares step, is replaced by the Bregman-divergence identity \eqref{HELLO}. This identity measures deviations from the optimal pair $(u_k^\ast,\hat w_k)$—the closed-form central controller and its corresponding worst-case disturbance—and parameterizes all level-$\gamma$ policies. Corollary~\ref{corr:1} confirms that Theorem~\ref{thm:main} is a direct generalization of standard $H_\infty$ control.

\begin{corollary}\label{corr:1}
Under quadratic costs $q(x)=x^\top Qx$, $r(u)=u^\top Ru$, $s(w)=w^\top Sw$, \eqref{eq:RICC2} yields the standard $H_\infty$ ARE and \eqref{eq:concave} reduces to the classical negativity test. The resulting controller~\eqref{eq:optimal_controller} is the standard linear $H_\infty$ feedback law.
\end{corollary}

\begin{remark}\label{rem:def_m}
    As it will become evident in the discussion in Section \ref{sec: disc}, the function $m(\cdot)\!\!:=\! (p-q)(\cdot)$ will play a central role in the appropriate choices of $q,r,s$ for various applications. Thus, we rewrite the Riccati-like equation as:
    \begin{align}
        &\!\!\!p(x)=q(x)+m(x)\label{eq:ricc1}\\       &\!\!\!g^\ast(\xi)=p^\ast(\xi)+r^\ast(B^\top\xi)=m^\ast(A^\top\xi)+\gamma^2s^\ast(\gamma^{-2}\xi)\label{eq:ricc2}
    \end{align}
\end{remark}

\subsection{Proof of the Main Theorem}\label{sec:proofmain}
We present now the proof of Theorem \ref{thm:main}. 
\begin{remark}
    In this proof, we assume that the system is fully actuated, i.e., $d= n$ and $B$ is invertible. However, in Appendix \ref{app:B}, we show that the result also holds for under- and over-actuated systems ($n \neq d$). 
\end{remark}

First we state a key lemma used to transform the objective.

\begin{lemma}\label{lemm:det}
    Consider the deterministic system: $x_{k+1}=Ax_k+Bu_k.$
We aim to find a causal time-invariant control policy \(\{u_k\}\) that stationarizes the following \emph{indefinite} cost function, for every finite horizon \(T\ge 1\):
\begin{align}
    C_T=\tilde p(Ax_T+Bu_T)+\sum\limits_{k=0}^{T}\tilde r(u_k)+\tilde q(x_k)
\end{align}
\vspace{-2mm}

where a stationary point satisfies $\frac{\partial C_T}{ \partial u_k}=0,\quad 0\leq k\leq T$.
Here, without assumptions on convexity, $\tilde r$ and $\tilde q$ are the designer-specified control and state cost, and $\tilde p$ is the storage function or terminal cost, which isn't arbitrary, and rather arises from a fixed point equation.

Then, the cost $C_T$ can be re-written in terms of Bregman divergence (in a purely algebraic sense since $\tilde r$ and $\tilde p$ need not be convex) as:
\begin{align}
  \!\! \!\!\! C_T=\sum\limits_{k=0}^{T}\! D_{\tilde r}(u_k,\hat{u}_k)\!+\!D_{\tilde 
 p}\Bigl(Ax_k+Bu_k,Ax_k+B\hat{u}_k\Bigr)
\end{align}
where $\hat{u}_k$ is a stationary point satisfying:
\begin{align}
   \nabla \tilde r(\hat{u}_k)+B^\top \nabla \tilde p(Ax_k+B\hat{u}_k)=0 
\end{align}
and where $\tilde p$ satisfies the Riccati-like equation
\begin{align}\label{eq:ric1}
    &\tilde p(x_k)=\tilde q(x_k)+\tilde r(\hat{u}_k)+\tilde p(Ax_k+B\hat{u}_k)\\
    &\text{or equivalently: }\tilde p(x_k)=\tilde q(x_k)+\tilde m(x_k),\\
    &\text{with }\tilde m(x_k)=\tilde r(\hat{u}_k)+\tilde p(Ax_k+B\hat{u}_k)
\end{align}
\end{lemma}
\begin{proof}
    The proof can be seen from repeatedly applying the identity:    $\tilde r(u_k)+\tilde q(x_k)+\tilde p(Ax_k+Bu_k)=\tilde p(x_k)+D_{\tilde r}(u_k,\hat{u}_k)+D_{\tilde p}(Ax_k+Bu_k,Ax_k+B\hat{u}_k)$,
    which can be verified from the definition of Bregman divergence and using the Riccati-like equation \eqref{eq:ric1}. 
\end{proof}
%\vspace{-2mm}

Now we proceed with the proof of Theorem \ref{thm:main}.
\vspace{3pt}
\subsubsection{Proof of the necessary and sufficient condition \eqref{eq:concave}}
We use Lemma \ref{lemm:det} to transform \eqref{eq:breg_cost-Inf2} into the following equivalent form:
\begin{align}\label{eq:bregdivcost}
    J_T=&\sum\limits_{k=0}^{T}D_{\gamma^2 s}(w_k,\hat{w_k})-D_r(u_k,\hat{u_k})\nonumber\\&-D_p\Bigl(Ax_k+Bu_k+w_k,Ax_k+B\hat{u}_k+\hat{w}_k\Bigr)
\end{align}
where $(\hat{u}_k(x_k),\hat{w}_k(x_k))$ is a stationary point satisfying:
\begin{align}
    &\nabla r(\hat{u}_k)+B^\top \nabla p\Bigl(Ax_k+B\hat u_k+\hat w_k\Bigr)=0 \label{eq:st1}\\
    &-\gamma^2\nabla s(\hat w_k)+\nabla p\Bigl(Ax_k+B\hat u_k+\hat w_k\Bigr)=0\label{eq:st2}
\end{align}
and $p$ satisfies the Riccati-like equation:
\begin{align} \label{eq:riccati-like}
     p(x_k)\!=\!q(x_k)\!+\!r(\hat u_k)\!-\!\gamma^2s(\hat w_k)\!+\!p\Bigl(\!Ax_k\!+\!B\hat u_k\!+\!\hat w_k\!\Bigr),
\end{align}
%\vspace{-5mm}
\noindent or equivalently,
\begin{align}
    &p(x_k)=q(x_k)+m(x_k),\label{eq:ppp}\\
    & m(x_k)=r(\hat u_k)-\gamma^2s(\hat w_k)+p(Ax_k+B\hat u_k+\hat w_k).\label{eq:m}
\end{align}
We seek a strictly convex terminal cost $p$ that satisfies the Riccati-like equation \eqref{eq:riccati-like}. Note that $\hat u_k$ and $ \hat w_k$ can be found as a function of $x_k$ using the conditions \eqref{eq:st1},\eqref{eq:st2}. Now, in the cost \eqref{eq:bregdivcost}, we will focus on the term $D_r(u_k,\hat{u}_k)+D_p\Bigl(Ax_k+Bu_k+w_k,Ax_k+B\hat{u}_k+\hat{w}_k\Bigr)$. First we use the Duality property of Bregman divergence in Lemma \ref{lemm:Breg}  to obtain:
\begin{align*}   &D_r(u_k,\hat{u}_k)+D_p\Bigl(Ax_k+Bu_k+w_k,Ax_k+B\hat{u}_k+\hat{w}_k\Bigr)\\
    &=D_{r^\ast}(\nabla r(\hat u_k), \nabla r(u_k))\\&+D_{p^\ast}\Bigl(\nabla p(Ax_k+B\hat{u}_k+\hat{w}_k), \nabla p(Ax_k+Bu_k+w_k)\Bigr).
\end{align*}
Next, we use the existence of the inverse of $B$ to continue with the manipulation: 
\begin{align*}
    &=D_{r^\ast(B^\top(\cdot))}(B^{-\top}\nabla r(\hat u_k), B^{-\top}\nabla r (u_k))\\&+ D_{p^\ast}\Bigl(\nabla p(Ax_k+B\hat{u}_k+\hat{w}_k), \nabla p(Ax_k+Bu_k+w_k)\Bigr),
\end{align*}
and leverage the definition of $\hat{u}_k$ from \eqref{eq:st1}:
\begin{align*}
    &=D_{r^\ast(B^\top(\cdot))}\Bigl(B^{-\top}\nabla r(\hat u_k), B^{-\top}\nabla r (u_k)\Bigr)\\&+ D_{p^\ast}\Bigl(-B^{-\top}\nabla r(\hat{u}_k), \nabla p(Ax_k+Bu_k+w_k)\Bigr).
\end{align*}
Since $p$ is even, it can be seen that $p^\ast$ is also even, therefore:
\begin{align*}
    &=D_{r^\ast(B^\top(\cdot))}\Bigl(B^{-\top}\nabla r(\hat u_k), B^{-\top}\nabla r (u_k)\Bigr)\\ &+ D_{p^\ast}\Bigl(B^{-\top}\nabla r(\hat{u}_k), -\nabla p(Ax_k+Bu_k+w_k)\Bigr).
\end{align*}
We leverage the completion-of-squares identity in Lemma \ref{lemm:Breg} to transform the above objective into the following:
    \begin{align*}
    &=D_{g_f} \Bigl(B^{-\top}\nabla r(\hat u_k),-\nabla  g_f^\ast(Ax_k+w_k)\Bigr) \nonumber\\&+D_{r^\ast(B^\top(\cdot))}\Bigl(-\nabla  g_f^\ast(A_kx+w_k),B^{-\top}\nabla r(u_k)\Bigr)\nonumber\\
    &+D_{p^\ast}\Bigl(-\nabla  g_f^\ast(Ax_k+w_k),-\nabla p (Ax_k+Bu_k+w_k)\Bigr),%\label{eq:bregg1}
\end{align*}
where $g_f(\cdot):= r^\ast(B^\top(\cdot))+p^\ast(\cdot)$. Let $g:\R^n \rightarrow \R$ be such that $g_f = g^\ast$. Continuing, we obtain:
\begin{align*}   &D_r(u_k,\hat{u}_k)+D_p\Bigl(Ax_k+Bu_k+w_k,Ax_k+B\hat{u}_k+\hat{w}_k\Bigr)\\&=D_{g^\ast} \Bigl(B^{-\top}\nabla r(\hat u_k),-\nabla  g(Ax_k+w_k)\Bigr)\\&+D_{r^\ast(B^\top(\cdot))}\Bigl(-\nabla  g(Ax_k+w_k),B^{-\top}\nabla r(u_k)\Bigr)\\
    &+D_{p^\ast}\Bigl(-\nabla  g(Ax_k+w_k),-\nabla p(Ax_k+Bu_k+w_k)\Bigr).
\end{align*}
Now, using the Duality property in Lemma \ref{lemm:Breg}, we get:
\begin{align*}
     &=D_{g} \Bigl(\nabla g^\ast(-\nabla  g(Ax_k+w_k)),\nabla g^\ast(B^{-\top}\nabla r(\hat u_k))\Bigr)\\
    &\!+\!\!D_{\!r(\!B^{\!-1}\!(\cdot)\!)}\!\biggl(\!\!B\nabla\! r^\ast\!\Bigl(\!B^\top \!B^{-\top}\!\nabla\!  r(u_k)\!\!\Bigr)\!,\!B\nabla \!r^\ast\!\Bigl(\!\!-\!B^\top\!\nabla  \!g(\!Ax_k\!\!+\!\!w_k\!)\!\!\Bigr)\!\!\!\biggr)\!\!\!\!\!\\
     &\!+\!\!D_{\!p}\!\biggl(\!\nabla\! p^\ast\!\Bigl(-\nabla p(Ax_k\!+\!Bu_k\!+\!w_k)\Bigr),\! \nabla\! p^\ast\Bigl(\!-\nabla \! g(Ax_k\!+\!w_k)\!\Bigr)\!\!\biggr)\!\!\!\!
\end{align*}
By the definition of the Fenchel dual, we have:
\begin{align}
    &=D_{g} \Bigl(-(Ax_k+w_k),\nabla g^\ast(B^{-\top}\nabla r(\hat u_k))\Bigr)\nonumber\\&+D_{r(B^{-1}(\cdot))}\biggl(Bu_k,B\nabla r^\ast\Bigl(-B^\top\nabla  g(Ax_k+w_k)\Bigr)\biggr)\nonumber\\
    &+D_{p}\biggl(-\Bigl(Ax_k+Bu_k+w_k\Bigr),\nabla p^\ast\Bigl(-\nabla  g(Ax_k+w_k)\Bigr)\biggr). \label{eq:breg_inter}
\end{align}

To further simplify this identity, we will use Lemma \ref{lem:aux} below.

\begin{lemma} \label{lem:aux}
    (Auxiliary results) The following holds:
    \begin{enumerate}
        \item $ \nabla p(Ax_k + Bu_k^\ast+ w_k) = \nabla g(Ax_k+w_k)$
        \item $ \nabla r(u_k^\ast)+B^\top\nabla p(Ax_k+Bu_k^\ast+w_k)=0 $
        \item $\nabla r(\hat{u}_k) + B^\top \nabla g(Ax_k + \hat w_k)=0$
    \end{enumerate}
\end{lemma}
\begin{proof}
    Recall the definition of $g_f$ as $g_f(\xi):= r^\ast(B^\top\xi)+p^\ast(\xi)$. Then, taking derivative of both sides with respect to $\xi$ yields:
    $\nabla g^\ast(\xi) = B \nabla r^\ast(B^\top \xi) + \nabla p^\ast(\xi).$
    Plugging for $\xi = \nabla g(Ax_k+w_k)$ and defining  
        $u^\ast_k:= -\nabla r^\ast\Bigl(B^\top\nabla g(Ax_k+w_k)\Bigr)$
    yields $Ax_k + Bu^\ast_k + w_k = \nabla p^\ast\Bigl(\nabla g(Ax_k+w_k)\Bigr)$.
    
    Since $p^\ast$ is strictly convex, $\nabla p^\ast$ is invertible and we obtain the first part of Lemma \ref{lem:aux}:
        $\nabla p\Bigl(Ax_k + Bu_k^\ast+ w_k\Bigr) = \nabla g(Ax_k+w_k).$
    Combining this with the definition of $u^\ast_k$ and using the invertibility of $\nabla r^\ast$, we get:
    \begin{align}
        \nabla r(u_k^\ast)+B^\top\nabla p\Bigl(Ax_k+Bu_k^\ast+w_k\Bigr)=0, \label{eq:rond}
    \end{align}
    which is the second statement of Lemma \ref{lem:aux}. 
    
    Note that since $r$ and $p$ are convex, \eqref{eq:rond} admits a unique solution which is $u^\ast(x_k,w_k)$. In fact, $u^\ast_k$ is the unique solution to the following convex optimization problem:
    \begin{align*}
        \min_{v_k} r(v_k) + p \Bigl( Ax_k + Bv_k + w_k\Bigr).
    \end{align*}
    This implies that $\hat u_k = u^\ast_k(x_k, \hat w_k)$. Therefore, the third statement follows by combining the first and second statements. This concludes the proof of Lemma \ref{lem:aux}. 
\end{proof}

\vspace{3pt}
Using Lemma \ref{lem:aux}, we arrive at
\begin{align}
    &\nabla g^\ast(B^{-\top}\nabla r(\hat u_k)) = - Ax_k + \hat{w}_k \label{eq:interm1} \\ 
    &B\nabla r^\ast\Bigl(-B^\top\nabla  g(Ax_k+w_k)\Bigr) = Bu_k^\ast \label{eq:interm2}\\ \nabla p^\ast(-&\nabla  g(Ax_k+w_k)) = -(Ax_k + Bu^\ast_k + w_k). \label{eq:interm3}
\end{align}
We plug identities \eqref{eq:interm1},\eqref{eq:interm2},\eqref{eq:interm3} into \eqref{eq:breg_inter}. After using the fact that $r,p,g$ are even, we arrive at following equivalent formulation of the cost $J_T$ \eqref{eq:bregdivcost}:
\begin{align} 
    &J_T= \sum_{k=0}^\top  D_{\gamma^2s}(w_k,\hat w_k)-D_{g} \Bigl(Ax_k+w_k,Ax_k+\hat w_k\Bigr)\nonumber\\ &-\!D_{r}(u_k,u_k^\ast)\! -\!D_{p}\Bigl(Ax_k\!+\!Bu_k+w_k,Ax_k+Bu_k^\ast+w_k\Bigr). \label{eq:finaljt}
\end{align}
%\vspace{-4mm}

Since $u_k$ cannot influence the first and second term of the equation \eqref{eq:finaljt}, a necessary condition for positivity is:
\begin{align} \label{cond:breg_pos}
    D_{\gamma^2s}(w_k,\hat w_k)-D_{g} \Bigl(Ax_k+w_k,Ax_k+\hat w_k\Bigr)\geq 0 \quad \forall w_k,  
\end{align}
which holds if and only if 
\begin{align}
    \gamma^2s(\cdot)-g(Ax_k+(\cdot))\quad \text{is convex in }(\cdot).\label{eq:nec}
\end{align} 
This argument establishes the necessesity of the condition presented in \eqref{eq:concave}. This condition is also sufficient because we can choose $u_k=u_k^\ast$. %\qed
\subsubsection{Derivation of the Riccati-like equation \eqref{eq:RICC2}}
In order to derive the Riccati-like equation in \eqref{eq:RICC2}, we perform further manipulation of the necessary and sufficient condition \eqref{cond:breg_pos}. We observe that \eqref{cond:breg_pos} is equivalent to: $-\gamma^2s(\hat w_k)+g(Ax_k+\hat w_k)  \geq \max_{w_k} -\gamma^2s(w_k)+g(Ax_k+w_k).$
Therefore, a necessary condition is that the stationary point $\hat w_k$ \eqref{eq:st1} \eqref{eq:st2} is also a global maximum:
\begin{align} \label{eq:globalmax}
    \!\!-\gamma^2s(\hat w_k)\!+\!g(Ax_k&+\hat w_k)\!=\! \max_{w_k} -\gamma^2s(w_k)\!+\!g(Ax_k+w_k).
\end{align}
Recall that $g^\ast (\cdot) =  r^\ast(B^\top(\cdot))+p^\ast(\cdot)$. By strict convexity of $r$ and $p$, we observe that
\begin{align}
    g(Ax_k + w_k) &= \min_{v_k} r(v_k) + p \Bigl( Ax_k + Bv_k + w_k\Bigr) \nonumber \\ &= r(u^\ast_k) + p \Bigl( Ax_k + Bu^\ast_k + w_k\Bigr). \label{eq:grp}
\end{align}
By uniqueness of $u^\ast_k$ from Lemma \ref{lem:aux}, we have that
\begin{align}\label{eq:sudden}
    g(Ax_k + \hat w_k) = r(\hat u_k) + p \Bigl( Ax_k + B\hat u_k +\hat w_k\Bigr).
\end{align}
We note that the original Riccati-like equation \eqref{eq:riccati-like} can be written as 
\begin{align}\label{eq:ric_mod}
    p(x_k) - q(x_k) = r(\hat{u}_k) - \gamma^2 s(\hat{w}_k) + p\Bigl(Ax_k + B\hat{u}_k + \hat{w}_k\Bigr).\!
\end{align}
Therefore, combining \eqref{eq:globalmax}, \eqref{eq:grp}, and \eqref{eq:ric_mod}, condition \eqref{cond:breg_pos} can be rewritten as 
\begin{align} 
    (p-q) (x) = \max_w -\gamma^2s(w) + g(Ax+w) \quad \forall x. \label{eq:mgs}
\end{align}

Hence we need to understand $(p-q)(\cdot)$ better. For the ease of notation, recall that in \eqref{eq:ppp} we defined $m(x):= p(x) - q(x)$. 
Taking derivative of both sides of \eqref{eq:ric_mod} with respect to $x_k$ yields: $\nabla m(x_k) = \nabla \hat{u}_k^\top \nabla r(\hat{u}_k) - \nabla \hat{w}_k^\top \gamma^2 \nabla s(\hat{w}_k) + \Bigl(A^\top + \nabla \hat{u}_k^\top B^\top  + \nabla \hat{w}_k\Bigr) \nabla p\Bigl(Ax_k + B \hat{u}_k + \hat{w}_k\Bigr).$
Now, using the previous stationary conditions \eqref{eq:m}, we have:
\begin{align} \label{eq:grdmp}
    \nabla m(x_k) = A^\top \nabla p \Bigl(Ax_k + B \hat{u}_k + \hat{w}_k\Bigr).
\end{align}
%\vspace{-3mm}
From Lemma \ref{lem:aux}, we have that that $B^\top \nabla p \Bigl(Ax_k + B \hat{u}_k + \hat{w}_k\Bigr) = - \nabla r(\hat{u}_k)$. Using \eqref{eq:st1}, this implies 
%\begin{align*}
    $B^\top A^{-\top} \nabla m(x_k) = -  \nabla r (\hat{u}_k),$
%\end{align*}
which, from invertibility of $\nabla r$, entails:
\begin{align} \label{eq:uhatrm}
    \hat{u}_k = - \nabla r^\ast\Bigl(B^\top A^{-\top} \nabla m(x_k)\Bigr)
\end{align}
Moreover, similarly using the stationary equation \eqref{eq:st2} yields
\begin{align}  \label{eq:whatsm}
    \hat{w}_k = \nabla s^\ast\Bigl(\gamma^{-2} A^{-\top} \nabla m(x_k)\Bigr).
\end{align}
Going back to \eqref{eq:grdmp}, using the invertibility of $\nabla p$, we obtain:
\begin{align}\label{eq:pinvax}
    \nabla p^\ast \Bigl( A^{-\top} \nabla m(x_k)\Bigr) = Ax_k + B\hat u_k + \hat w_k .
\end{align}
Plugging in \eqref{eq:uhatrm} and \eqref{eq:whatsm} into \eqref{eq:pinvax}, we arrive at the following:
$\nabla p^\ast \Bigl( A^{-\top} \nabla m(x_k)\Bigr) = Ax_k - B \nabla r^\ast \Bigl(B^\top A^{-\top} \nabla m(x_k)\Bigr) + \nabla s^\ast\Bigl(\gamma^{-2} A^{-\top} \nabla m(x_k)\Bigr).$
Note that the above identity is only a function of $A^{-\top} \nabla m(x_k)$ and $x_k$. To understand $\nabla m(x_k)$, we consider \eqref{eq:mgs}. The objective of the optimization is strictly concave in $w$ and strictly convex in $x$. This implies $m(\cdot)$ will be strictly convex in $x$.  Denoting $\xi:=A^{-\top} \nabla m(x_k)$, we observe that by strict convexity of $m$, $\nabla m$ is invertible, therefore yielding: $p^\ast(\xi)\! + \!r^\ast(B^\top \xi) \!=\! A \nabla m^\ast (A^\top \!\xi) \!-\! B \nabla r^\ast ( B^\top\! \xi) \!+\! \nabla s^\ast(\gamma\!^{-2} \xi).$
All in all, by Assumption \ref{asum:q} on $p,r,s$, we obtain:
\begin{align*}
    p^\ast(\xi) + r^\ast(B^\top \xi) = (p-q)^\ast(A^\top \xi) + \gamma^2 s^\ast (\gamma^{-2}\xi).
\end{align*}
\vspace{3pt}
This concludes the derivation of \eqref{eq:RICC2}. 

\subsubsection{Lyapunov stability}
The function \(p(\cdot)\) acts as a Lyapunov function. 
By combining equations \eqref{eq:ppp} and \eqref{eq:grp}, we obtain: $
   p\Bigl(Ax_k+Bu_k^\ast+w\Bigr)-p(x_k)=g(Ax_k+w_k) - r(u^\ast_k) - \Bigl(q(x_k)+m(x_k)\Bigr).$ 
Then using \eqref{eq:mgs}, we have that for every $\{w_k\}\neq 0$,  $  p\Bigl(Ax_k+Bu_k^\ast+w_k\Bigr)-p(x_k)  \leq m(x_k)+\gamma^2 s(w_k)- r(u_k^\ast) - q(x_k)-m(x_k) =\gamma^2 s(w_k)- r(u_k^\ast) - q(x_k),$ which ensures bounded-input bounded-output stability when $w_k\neq0$.
And for $w_k=0$, we arrive at: $p(Ax_k+Bu_k^\ast)-p(x_k)\leq - r(u_k^\ast) - q(x_k)$
which is negative, thus ensuring Lyapunov stability. 
\subsubsection{Proof of Corollary \ref{corr:1}} see Appendix \ref{app:corr}. 

Finally, we note that while this proof derives a full-information controller $u_k=f(x_k,w_k)$, the extension to state-feedback $u_k=f(x_k)$ is straightforward.

\section{Discussion on the Controller Synthesis}\label{sec: disc}
\textbf{Theorem \ref{thm:main}} gives an explicit stabilizing \emph{nonlinear} full-information controller for a \emph{linear} system when minimizing strictly convex (nonquadratic) costs. 

In practice, the state, control, and disturbance cost functions ($q,r,s$) must adhere to the conditions specified in Theorem \ref{thm:main}. Ideally, the designer selects \(q\), \(r\), and \(s\) and then solves the Riccati-like equation \eqref{eq:RICC2} for \(p\) while ensuring \eqref{eq:concave} holds. However,  obtaining a valid strictly convex $p$ for given \(q\), \(r\), and \(s\) is often challenging. To overcome this difficulty, we propose an \emph{offline} synthesis procedure which retains a workflow close to classical $H_\infty$.

We propose three complementary design options; in each, the designer begins by fixing the structure of either $m(\cdot)$ or $g(\cdot)$. For simplicity, in the first two approaches, we fix $m(\cdot)$ to a quadratic form
$m(x) = x^\top M x,$
with $M \succ 0$, ensuring its strict convexity. In the third approach, we instead fix $g(\cdot)$ to a quadratic form,
$g(x) = x^\top G x,$ with $G \succ 0$. The designer then proceeds with one of the following options.
\begin{enumerate}
\item \textbf{State Cost and Disturbance Specification:} The designer selects the state cost $q(\cdot)$ and the disturbance cost $s(\cdot)$. For a given $M\succ 0$, the control cost $r(\cdot)$ is then determined using Equations~\eqref{eq:ricc1} and \eqref{eq:ricc2}. In this setting, \(M\) must be chosen carefully to ensure that the resulting \(r(\cdot)\) is both convex and positive. Thus, the design problem reduces to finding an \(M \succ 0\) that satisfies these criteria, a task that can be addressed using (possibly convex) optimization techniques. \label{a:1}

 \item \textbf{Control Cost and Disturbance Cost Specification:} Here, the designer chooses the control cost $r(\cdot)$  and disturbance cost $s(\cdot)$. With $M$ specified, the state cost $q(\cdot)$ is computed via the same equations. The challenge lies in finding an $M\succ 0$ that ensures the resulting $q(\cdot)$ is convex and positive, again potentially addressed through (convex) optimization.\label{a:2}

 \item \textbf{Control Cost and State Cost Specification:} In this approach, the designer fixes the control cost $r(\cdot)$ and state cost $q(\cdot)$. With a selected $G\succ0$, the disturbance cost $s(\cdot)$ is then derived. The design problem is to choose $G$ such that the computed $s(\cdot)$ remains convex and positive, which can similarly be achieved using optimization techniques. \label{a:3}
\end{enumerate}
In all approaches, the designer must also ensure that $-\gamma^2s(\cdot)+g(Ax+(\cdot))$ is concave in $(\cdot)$, as stated in Theorem \ref{thm:main}. Once feasibility holds, the designer implements the closed-form controller \eqref{eq:optimal_controller}.

In short, this synthesis retains the classical advantages (one offline step, closed-form law, no online optimization) while accommodating richer specifications. It enables principled safety envelopes, bang--bang or saturating actuation, and sparse inputs by appropriate choices of $(q,r,s)$ and the shaping functions $(m,g)$. 

%%%%%%%%%%%%%%

Next, we present all the different Approaches \ref{a:1},\ref{a:2},\ref{a:3} in details.
%\vspace{-3mm}

\subsection{Designing the control cost \texorpdfstring{$r(\cdot)$}{r(·)} and the disturbance cost \texorpdfstring{$s(\cdot)$}{s(·)}:}
 In this approach, the designer selects even, positive, and strictly convex cost functions \(r(\cdot)\) and \(s(\cdot)\), and fixes a structure of the function $m(\cdot)$, assumed to be $ m(x)=x^\top Mx$, where $M\succ 0$ is a free parameter. 
Theorem~\ref{thm:chooseR} establishes necessary and sufficient conditions for a positive definite matrix \(M\) to exist such that $q(\cdot)$ is convex and positive, and such that the necessary and sufficient conditions of Theorem \ref{thm:main} are satisfied. This guarantees the existence of the optimal controller as stated in Equation~\eqref{eq:optimal_controller}.

\begin{theorem}\label{thm:chooseR}
    Let \( r: \mathbb{R}^d \to \mathbb{R} \) and \( s: \mathbb{R}^n \to \mathbb{R} \) be convex, even, and positive functions chosen by the designer. Consider a stable matrix \( A \). Then, for a given \(M\succ 0\), a causal time-invariant control policy that solves Problem \ref{pb:bregman} exists, i.e., the functions \(p(\cdot)\) and $q(\cdot)$ given by
    %\vspace{-0.68mm}    
\begin{align}
   & p^\ast(\xi) =\! -r^\ast \left( B^\top \xi \right) \!+\! \frac{1}{4} \xi^\top A M^{-1} A^\top \xi\!+\!\gamma^2 s^\ast(\gamma^{-2}\xi),\label{eq:UGH}  \\
    &q(x) = p(x) - x^\top M x \label{eq:UGH2}
\end{align}
are convex and positive, and the function $-\gamma^2 s(\cdot)+g(Ax+(\cdot))$ is concave in $(\cdot)$ if and only if \(M\) satisfies the conditions:   
    \begin{align}
        &L\preceq B \nabla^2 r^*\left( B^\top \xi \right) B^\top -\gamma^{-2}\nabla^2 s^\ast(\gamma^{-2}\xi)
        \preceq U,\text{ }\forall \xi\label{eq:pqConvex} \\
        &\text{with }L:=\frac{1}{2} \left( A M^{-1} A^\top - M^{-1} \right), \quad U:= \frac{1}{2} A M^{-1} A^\top, \nonumber \\
        %&M^{-1}\succ 0\label{eq:MMM}\\
        &\text{and for a fixed $x$, }
        %&\quad\quad\quad\quad\quad\quad\quad\quad 
        \gamma^{-2}\nabla^2 s^\ast \Bigl(\gamma^{-2}\nabla g(Ax+w)\Bigr)\label{eq:MMM}\\
        &\quad\quad\quad\quad\quad\quad+\frac{1}{2}AM^{-1}A^\top\succeq \gamma^{-2}\nabla^2 s^\ast (\nabla s(w)) \quad \forall w.\nonumber
    \end{align}

  Furthermore, suppose that $r(\cdot)$ and $s(\cdot)$ are strongly convex and smooth satisfying
    \begin{equation}\label{eq:rCvxSmooth}
   L_r\preceq \nabla^2r\preceq U_r,\quad L_s\preceq\nabla^2s\preceq U_s, %\text{ or }U_r^{-1}\preceq\nabla^2r^\ast(\cdot)\preceq L_r^{-1}
    \end{equation}
    for some positive definite matrices \( L_r,U_r \succ 0 \), and \( L_s,U_s \succ 0 \). Then, for any matrix \( A \), for a given \(M\succ 0\), a time-invariant control policy that solves Problem \ref{pb:bregman} exists, if $M^{-1}$ satisfies the following convex feasibility program:
    \begin{equation}\label{eq:cvxopt}
    \begin{aligned}
        %U \preceq \frac{1}{4} A M^{-1} A^\top \preceq L + M^{-1}.
        &\frac{1}{2} \left( A M^{-1} A^\top - M^{-1} \right) 
        \preceq BU_r^{-1} B^\top -\gamma^2L_s^{-1} \\
        &BU_r^{-1} B^\top -\gamma^2L_s^{-1} \succ 0\\
        & B L_r^{-1} B^\top \preceq \gamma^2U_s^{-1}
        + \frac{1}{2} A M^{-1} A^\top\\  
        &\gamma^{-2}L_s^{-1}\preceq \frac{1}{2}AM^{-1}A^\top+\gamma^{-2}U_s^{-1}.
    \end{aligned}           
    \end{equation}
\end{theorem}
\begin{proof}
 The proof is in Appendix \ref{app:1},\ref{app:2}.
\end{proof}
%\vspace{-2mm}
In the case where $A$ is stable, the designer enjoys greater flexibility in selecting the cost function $r$ which need not be strongly convex. In this case, Theorem~\ref{thm:chooseR} reduces the design problem to finding a positive definite matrix \(M\) that satisfies conditions \eqref{eq:pqConvex},\eqref{eq:MMM} which admit relaxations or convexification in structured cases (e.g., $s(w)=w^\top Sw$). 

In the case where $A$ is any matrix (stable or unstable), $r$ and $s$ are chosen to be any strongly convex and smooth, and thus Theorem~\ref{thm:chooseR} reduces the design problem to finding a positive definite matrix \(M\) that satisfies the convex program \eqref{eq:cvxopt}. Once \( M \) is obtained, the corresponding cost functions \( q(\cdot)\) and $p(\cdot)$ are computed from \eqref{eq:UGH}, \eqref{eq:UGH2}. Finally, to compute the optimal controller \eqref{eq:optimal_controller}, we need $\nabla g(\cdot)$ which can be found from: $\nabla g^\ast(\xi)=\frac{1}{2}AM^{-1}A^\top\xi+\nabla s^\ast(\gamma^{-2}\xi).$

\subsection{Designing the state cost \texorpdfstring{$q(\cdot)$}{q(.)} and the disturbance cost \texorpdfstring{$s(\cdot)$}{s(.)}} \label{sec:designQS}
In this approach, the designer selects even, positive, and strictly convex cost functions \(q(\cdot)\) and \(s(\cdot)\), and fixes a structure of the function $m(\cdot)$, assumed to be $ m(x)=x^\top Mx$, where $M\succ 0$ is a free parameter. 
Theorem~\ref{thm:chooseQ} establishes necessary and sufficient conditions for a positive definite matrix \(M\) to exist such that $r(\cdot)$ is convex and positive, and such that the necessary and sufficient conditions of Theorem \ref{thm:main} are satisfied. This guarantees the existence of the optimal controller as stated in Equation~\eqref{eq:optimal_controller}.

\begin{theorem}\label{thm:chooseQ}
    Consider a fully actuated or overactuated system, i.e., \(B\in \mathbb{R}^{n\times d}\) with \(d\geq n\). Let \( q: \mathbb{R}^n \to \mathbb{R} \) and \( s: \mathbb{R}^n \to \mathbb{R} \) be strictly convex, even, and positive functions chosen by the designer. Then, for a given \(M\succ 0\), a causal time-invariant control policy that solves Problem \ref{pb:bregman} exists, i.e., the function \(r(\cdot)\) 
   % \vspace{-3mm}    
\begin{align}\label{eq:UGH_r} 
r^\ast   \left(  \xi \right)   =&  -p^\ast(B^{- \dagger,\top}\xi)  +  \frac{1}{4}  \xi^\top  B^{-\dagger} A M^{-1} A^\top B^{-\dagger,\top}\xi \nonumber\\
&+\gamma^2 s^\ast(\gamma^{-2}B^{-\dagger,\top}\xi)
  % r^\ast  \! \!\left(  \xi \right)   \!\!= \!\! -p^\ast(B^{- \dagger,\top}\xi)  \!\!+ \!\! \frac{1}{4} \! \xi^\top  \!B^{-\dagger} \!\!A\! M^{-1}\!\! A^\top\! B^{-\dagger,\top}\!\xi\!\!+\!\!\gamma^2 s^\ast\!(\!\gamma^{\!-2}\!B^{-\dagger,\top}\!\xi)
\end{align}
is convex (and thus positive), and the function $\gamma^2 s(\cdot)-g(Ax+(\cdot))$ is convex in $(\cdot)$ if and only if \(M\) satisfies the conditions:   
    \begin{align}
        &[\nabla^2 q(x)+2M] [\frac{1}{2}AM^{-1}A^\top \\
        &\quad \quad \quad\quad\quad+\gamma^{-2}\nabla^2 s^\ast(\gamma^{-2}(\nabla q(x)+2M))]\succeq I,\quad \forall x \label{eq:rConvex}\\
        &\text{and for a fixed $x$, }
        %&\quad\quad\quad\quad\quad\quad\quad\quad 
        \gamma^{-2}\nabla^2 s^\ast \Bigl(\gamma^{-2}\nabla g(Ax+w)\Bigr)\label{eq:MMM2}\\
        &\quad\quad\quad\quad\quad\quad+\frac{1}{2}AM^{-1}A^\top\succeq \gamma^{-2}\nabla^2 s^\ast (\nabla s(w)) \quad \forall w.\nonumber
    \end{align}
  Furthermore, suppose that $q(\cdot)$ is strongly convex and $s(\cdot)$ is strongly convex and smooth satisfying
    \begin{align}
   &\nabla^2 q(\cdot)\succeq L_q, \quad  L_s\preceq \nabla^2 s(\cdot)\preceq U_s, %\text{ or, } U_s^{-1}\preceq\nabla^2 s^\ast(\cdot)\preceq L_s^{-1}%U_s=S^{-1}>= \gamma^2\nabla^2 s(\cdot)
    \end{align}
    for some positive definite matrices \( L_q \succ 0 \), and \( L_s,U_s \succ 0 \). Then, for a given \(M\succ 0\), a time-invariant control policy that solves Problem \ref{pb:bregman} exists, if $M$ satisfies the following feasibility program:
    \begin{equation}\label{eq:cvxopt2}
    \begin{aligned}
 %&\frac{1}{2}L_q+\gamma^{-2}L_q U_s^{-1} A^{-\top} MA^{-1}\\
 %&\quad\quad\quad+M+2MSA^{-\top} MA^{-1}\succeq A^{-\top} MA^{-1}\\
        % &  \gamma^{-2}L_s^{-1}\preceq \frac{1}{2}AM^{-1}A^\top+\gamma^{-2}U_s^{-1}\\
         &\frac{1}{2}  \gamma^{2}(L_s^{-1}-U_s^{-1})^{-1}\preceq A^{-\top}MA^{-1}\preceq Z_M\\
         &  Z_M\!\!:=\!\!\frac{1}{2}L_q\!+\!\gamma^{-2}L_q U_s^{-1} A^{-\top} MA^{-1}\!+\!M\!+\!2MSA^{-\top} MA^{-1}\!\!
    \end{aligned}           
    \end{equation}

\end{theorem}

\begin{proof}
 The proof idea is in the Appendix \ref{app:2}.
\end{proof}

In the case where $q$ and $s$ are chosen to be any strongly convex and smooth, Theorem~\ref{thm:chooseQ} reduces the design problem to finding a positive definite matrix \(M\) that satisfies the feasibility program \eqref{eq:cvxopt2} which can be further relaxed to a convex program. Once \( M \) is obtained, the corresponding cost function \( r(\cdot)\) is computed from \eqref{eq:UGH_r}. Finally, to compute the optimal controller \eqref{eq:optimal_controller}, we need $\nabla g(\cdot)$ which can be found from: $\nabla g^\ast(\xi)=\frac{1}{2}AM^{-1}A^\top\xi+\nabla s^\ast(\gamma^{-2}\xi).$

\subsection{Designing the state cost \texorpdfstring{$q(\cdot)$}{q(.)} and the control cost \texorpdfstring{$r(\cdot)$}{r(.)}}\label{sec:designQR}

In this approach, the designer selects even, positive, and strictly convex cost functions \(q(\cdot)\) and \(r(\cdot)\), and fixes a structure of the function $g(\cdot)$, assumed to be $ g(x)=x^\top Gx$, where $G\succ 0$ is a free parameter. 
Theorem~\ref{thm:chooseQR_} establishes necessary and sufficient conditions for a positive definite matrix \(G\) to exist such that $s(\cdot)$ is convex and positive, and such that the necessary and sufficient conditions of Theorem \ref{thm:main} are satisfied. This guarantees the existence of the optimal controller as stated in Equation~\eqref{eq:optimal_controller}.

\begin{theorem}\label{thm:chooseQR_}
    Consider a fully actuated or overactuated system, i.e., \(B\in \mathbb{R}^{n\times d}\) with \(d\geq n\). Let \( q: \mathbb{R}^n \to \mathbb{R} \) and \( r: \mathbb{R}^n \to \mathbb{R} \) be strictly convex, even, and positive functions chosen by the designer. Then, for a given \(G\succ 0\), a causal time-invariant control policy that solves Problem \ref{pb:bregman} exists, i.e., the function \(r(\cdot)\) 
   % \vspace{-3mm}    
\begin{align}\label{eq:UGH_s} 
   s^\ast(\xi)=\gamma^{-2}\left(p^\ast(\gamma^2\xi)+r^\ast(\gamma^2B^\top\xi)-m^\ast(\gamma^2A^\top\xi)\right)
\end{align}
is convex (and thus positive), and the function $\gamma^2 s(\cdot)-g(Ax+(\cdot))$ is convex in $(\cdot)$ if and only if \(G\) satisfies the conditions:   
    \begin{align}
    & \frac{1}{2}G^{-1}-B\nabla^2 r^\ast(B^\top \xi)B^\top \succeq 0\quad \forall \xi\label{pCVX}\\
    &\frac{1}{2}G^{-1}\preceq B\nabla^2 r^\ast(B^\top \nabla p(\xi))B^\top + \nabla^2 q^\ast (\nabla q(\xi))\quad \forall \xi\label{mCVX}\\
    & G+2A^\top GA B\nabla^2 r^\ast(\nabla p(\xi))B^\top G\nonumber\\
    &+\nabla^2 q(\xi)B\nabla^2 r^\ast(\nabla p(\xi))B^\top G\succeq A^\top GA +\frac{1}{2}\nabla^2 q(\xi)\quad \forall \xi\label{sCVX}
    \end{align}

  Furthermore, suppose that $q(\cdot)$ is smooth and strongly convex and $r(\cdot)$ is smooth satisfying
    \begin{align}
    L_q\preceq \nabla^2 q(\cdot)\preceq U_q, \quad \nabla^2 r(\cdot)\preceq U_r, \text{ or, } \nabla^2 r^\ast(\cdot)\succeq U_r^{-1}
    \end{align}
    for some positive definite matrices \( L_q,U_q \succ 0 \), and \( U_r \succ 0 \). Then, for a given \(G\succ 0\), a time-invariant control policy that solves Problem \ref{pb:bregman} exists, if $G$ satisfies the following feasibility program:
    \begin{equation}\label{eq:cvxopt3}
    \begin{aligned}  
        &G\preceq \frac{1}{2}B^{-\top,\dagger} \nabla^2 r(B^\top x)B^{-1} \quad \forall x\\
        &G\succeq \frac{1}{2}[BU_r^{-1}B^\top + U_q^{-1}]^{-1}\\
        &G\!+\!2A^\top GA BU_r^{-1}B^\top G\!+\!L_qBU_r^{-1}B^\top  G\succeq A^\top GA \!+\!\frac{1}{2}U_q
    \end{aligned}           
    \end{equation}

\end{theorem}

 \begin{proof}
 The proof idea is in the Appendix \ref{app:2}.
\end{proof}

In the case where $q$ is chosen to be any strongly convex and smooth function and $r$ any smooth function, Theorem~\ref{thm:chooseQR_} reduces the design problem to finding a positive definite matrix \(G\) that satisfies the feasibility program \eqref{eq:cvxopt3} which can be further relaxed to a convex program. Once \( G \) is obtained, the corresponding cost function \( s(\cdot)\) is computed from \eqref{eq:UGH_s}. Finally, to compute the optimal controller \eqref{eq:optimal_controller}, we need $\nabla g(\cdot)$ which is simply $\nabla g(x)=2Gx$.

\begin{remark}
In each approach, instead of merely solving a feasibility problem to obtain $M$ or $G$, the designer may opt to penalize the norm of $M$ or $G$, or employ an alternative cost function. Such modifications can significantly influence controller performance.
\end{remark}
\begin{remark}
    Note that Theorems \ref{thm:chooseQ} and \ref{thm:chooseQR_} assume a fully or overactuated system. The underactuated system case, i,e, $n>d$ requires more care and is left for future work.
\end{remark}
%\vspace{-2mm}

\section{Applications and Simulations}
We illustrate the framework on a scalar system $x_{k+1}=ax_k+bu_k+w_k$. We begin with an input-limited design. 

\subsection{Input-limited control through designing \texorpdfstring{$r(\cdot)$}{r(.)} and \texorpdfstring{$s(\cdot)$}{s(.)}} Many systems operate under hard per-step actuator limits (power, torque, duty cycle). We encode a strict amplitude budget \(t\) via
\begin{align}
r(u) &= 
\begin{cases}
u^2 & \text{if } |u| < t\\
  \infty & \text{if } |u| \geq t
\end{cases} 
\end{align}
For simplicity, we choose $s(w)=sw^2$, $s>0$, and fix $m(x)=mx^2$, for some $m$. The induced state cost becomes 
\begin{align}
q(x) &= \begin{cases}
v_\gamma(\frac{1}{a^2-b^2v_\gamma }-1)x^2, \quad\text{if } |x| \leq \frac{a^2t}{v_\gamma b}-tb,\\
   v_\gamma(\frac{1}{a^2}-1)x^2+\frac{2v_\gamma bt}{a^2}|x|+t^2(\frac{v_\gamma b^2}{a^2}-1),& \\\text{if } |x| > \frac{a^2t}{v_\gamma b}-tb
\end{cases}
\end{align}
with $v_\gamma= \frac{\gamma^2msa^2}{m+\gamma^2 a^2s}.$
The optimal controller becomes
\begin{align*}
    u^{\ast}(x) =
    \begin{cases}
    -\,\dfrac{b v_\gamma}{a} \, x, & \text{if } \left| x \right| \leq \dfrac{a t}{b v_\gamma}, \\[10pt]
    -\,a \cdot \text{sign}(x), & \text{if } \left| x \right| > \dfrac{a t}{b v_\gamma}.
    \end{cases}
\end{align*}
Figure~\ref{fig:bangR-scalar} considers the system with $a=0.6$ and $b=1$ and compares our controller (for $s=1$, $m=0.11$, and $t=0.1$) with the standard infinite-horizon $H_\infty$ controller (with unit weights). The left-hand-side figures show the state evolution and input control as a function of time under the worst-case $H_\infty$ disturbance. The right-hand-side figures do the same for the worst-case disturbance of our controller. In both cases, $\gamma = 1.32$ (chosen slightly above \(\gamma_{H_\infty}\)). Note that our controller maintains competitive performance even though the control signal is bounded by $t=0.1$. 

%\emph{Computational note.} Our controller is explicit at runtime after a single offline solve.

\begin{figure}[h]
  \centering
  \includegraphics[width=\columnwidth]{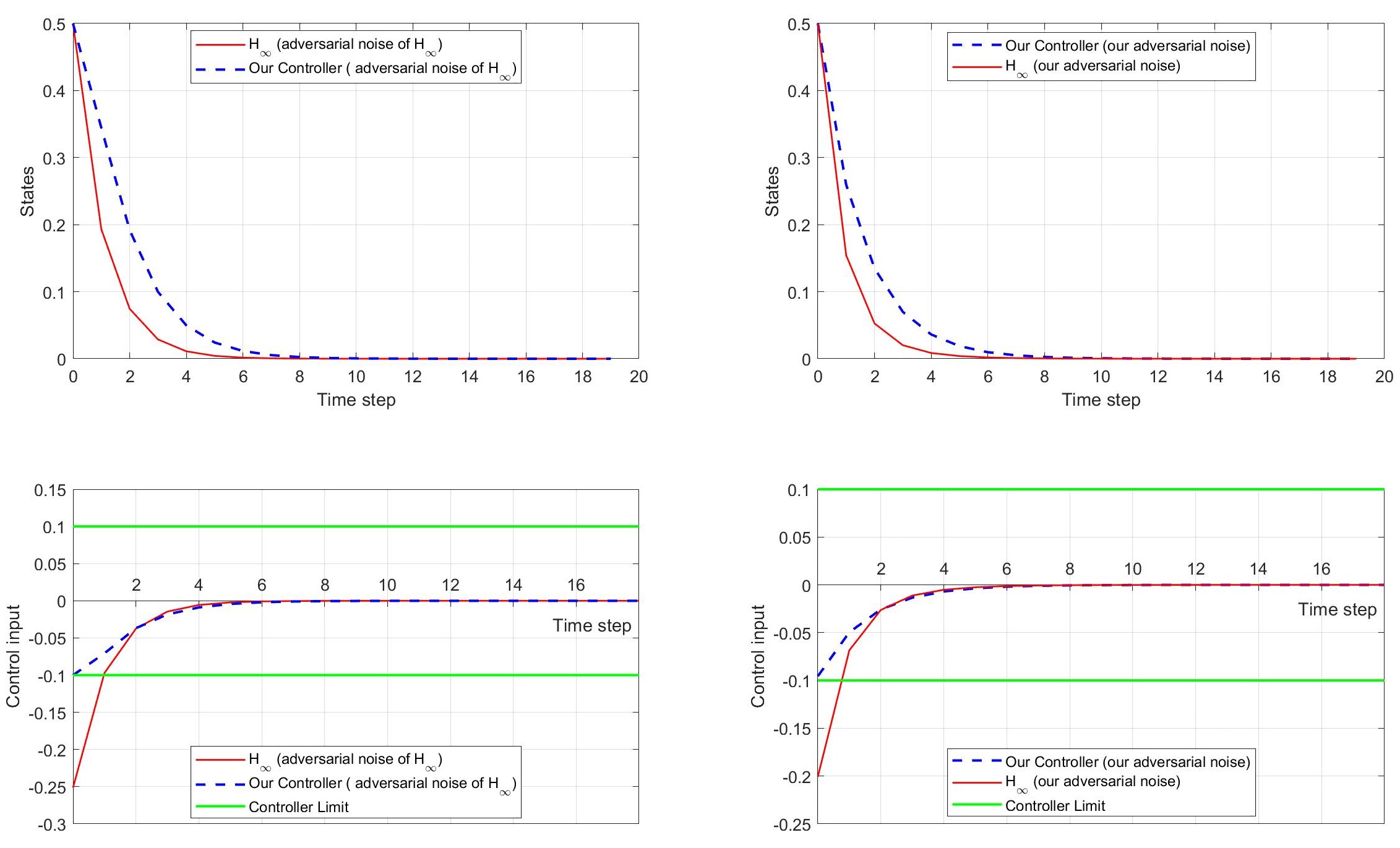}
  \caption{}%States and inputs with a saturating $r(\cdot)$ ($t=0.1$), under the worst-case noise of the quadratic $H_\infty$ controller (left) and of our controller (right).
  \label{fig:bangR-scalar}
\end{figure}
%\vspace{-4mm}
Next, we consider other designs, in particular, 1) the safety envelope ~
$
\!\!q(x)\!\!=\!\!\begin{cases}
x^2,& \!\!\!|x|<t\\
\infty,&\!\!\! |x|\ge t
\end{cases}
$, and 2) the exponential cost \(r(u)=e^{|u|}-|u|-1\).

\subsection{Safety control through designing \texorpdfstring{$q(\cdot)$}{q(.)} and \texorpdfstring{$s(\cdot)$}{s(.)}}
Many applications require a hard state envelope (e.g., safety corridors or keep-out zones). We encode a strict safety radius \(t>0\) via
\begin{align}
q(x) &=
\begin{cases}
x^2, & \text{if } |x|<t,\\
\infty, & \text{if } |x|\ge t~.
\end{cases}
\end{align}
For simplicity, we choose \(s(w)=s\,w^2\) with \(s>0\), and fix \(m(x)=m\,x^2\) for some \(m>0\). The induced input cost (from \eqref{eq:UGH_r}) becomes
\begin{align}
r(u) &= 
\begin{cases}
\frac{v_\gamma(v_\gamma+1)b^2}{(v_\gamma+1)a^2-v_\gamma}u^2 \quad \text{if } |u| < \frac{v_\gamma+1}{b v_\gamma} t a^2 - \frac{t}{b},\\
\frac{v_\gamma b^2}{a^2}(|u|+ \frac{t}{b})^2 - (v_\gamma+1)t^2\\
\text{if } |u| \ge \frac{v_\gamma+1}{b v_\gamma} t a^2 - \frac{t}{b},
\end{cases} 
\end{align}
with $v_\gamma= \frac{\gamma^2m sa^2}{m+\gamma^2 a^2s}.$

The optimal controller becomes
$$
u^\ast(x,w) =
\begin{cases}
\Bigl(-\dfrac{a}{b} \,  + \dfrac{v_\gamma}{a b (v_\gamma + 1)} \,\Bigr) (x+\frac{w}{a}) \\ \text{if } \left| \dfrac{2v_\gamma}{a} (x+\frac{w}{a}) \right| \leq 2 (v_\gamma + 1) t, \\[10pt]
-\dfrac{a}{b} \, (x+\frac{w}{a}) + \dfrac{t \cdot \text{sign}(x)}{b} \\ \text{if } \left| \dfrac{2v_\gamma}{a} (x+\frac{w}{a}) \right| > 2 (v_\gamma + 1) t.
\end{cases}
$$
Figure~\ref{fig:bangQ-scalar} considers the system with $a=1.1$ and $b=1$ and compares our controller (for $s=1$, $m=0.2$, and $t=0.2$) with the standard infinite-horizon $H_\infty$ controller (with unit weights). The left-hand-side figures show the state evolution and input control as a function of time under white noise. The mid figures so the same for uniform noise, and the right-hand-side figures do the same for Laplacian noise. In all cases, $\gamma=1.32$ (chosen slightly above \(\gamma_{H_\infty}\)). Note that our controller maintains competitive performance even though the state is kept within the hard safety envelope $\lvert x\rvert \le 0.2$.

\begin{figure}[h]
  \centering
  \includegraphics[width=\columnwidth]{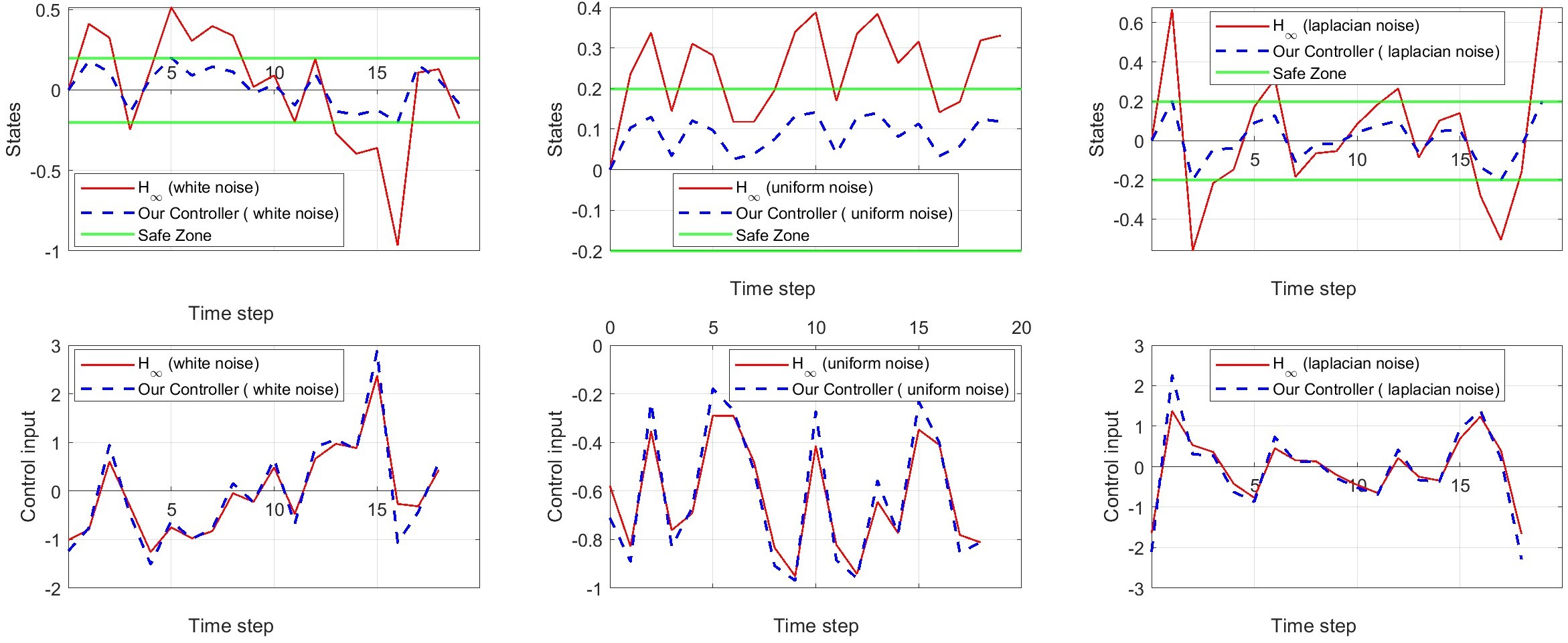}
  \caption{}
  \label{fig:bangQ-scalar}
\end{figure}

\subsection{Exponential Control through designing \texorpdfstring{$q(\cdot)$}{q(.)} and \texorpdfstring{$r(\cdot)$}{r(.)} }
Many systems operate in regimes where one wishes to use as little control effort as possible while still regulating the state. We encode a heavy penalty on the control through an exponential cost, while retaining a quadratic state penalty. Concretely, we set
\begin{align*}
r(u) &= e^{|u|}-|u|-1, \qquad
q(x)=x^2. 
\end{align*}
For simplicity, we fix \(g(x)=g\,x^2\) for some \(g>0\).
The optimal controller becomes
\begin{align*}
u^{\ast}(x,w) \;=\; -\,\operatorname{sign}\!\bigl(b(ax+w)\bigr)\,\log\!\Bigl( \bigl|\,2bg(ax+w)\,\bigr|+1 \Bigr).
\end{align*}
In this setting, the disturbance cost \(s(\cdot)\) and terminal cost \(m(\cdot)\) do not admit simple closed-form expressions, so the worst-case disturbance cannot be written in closed form.

Figure~\ref{fig:rEXP} considers the system with \(a=0.9\) and \(b=0.1\) and compares our controller (for \(g=15\)) with the standard infinite-horizon \(H_\infty\) controller (with unit weights). The left-hand-side figures show the state evolution and input control as a function of time under white noise. The mid figures so the same for uniform noise, and the right-hand-side figures do the same for Laplacian noise. In all cases, \(\gamma=7.45\)(chosen slightly above \(\gamma_{H_\infty}\)). Across all cases, the exponential penalty yields noticeably smaller control amplitudes while maintaining state trajectories comparable to the quadratic \(H_\infty\) baseline.

\begin{figure}[h]
  \centering
  \includegraphics[width=\columnwidth]{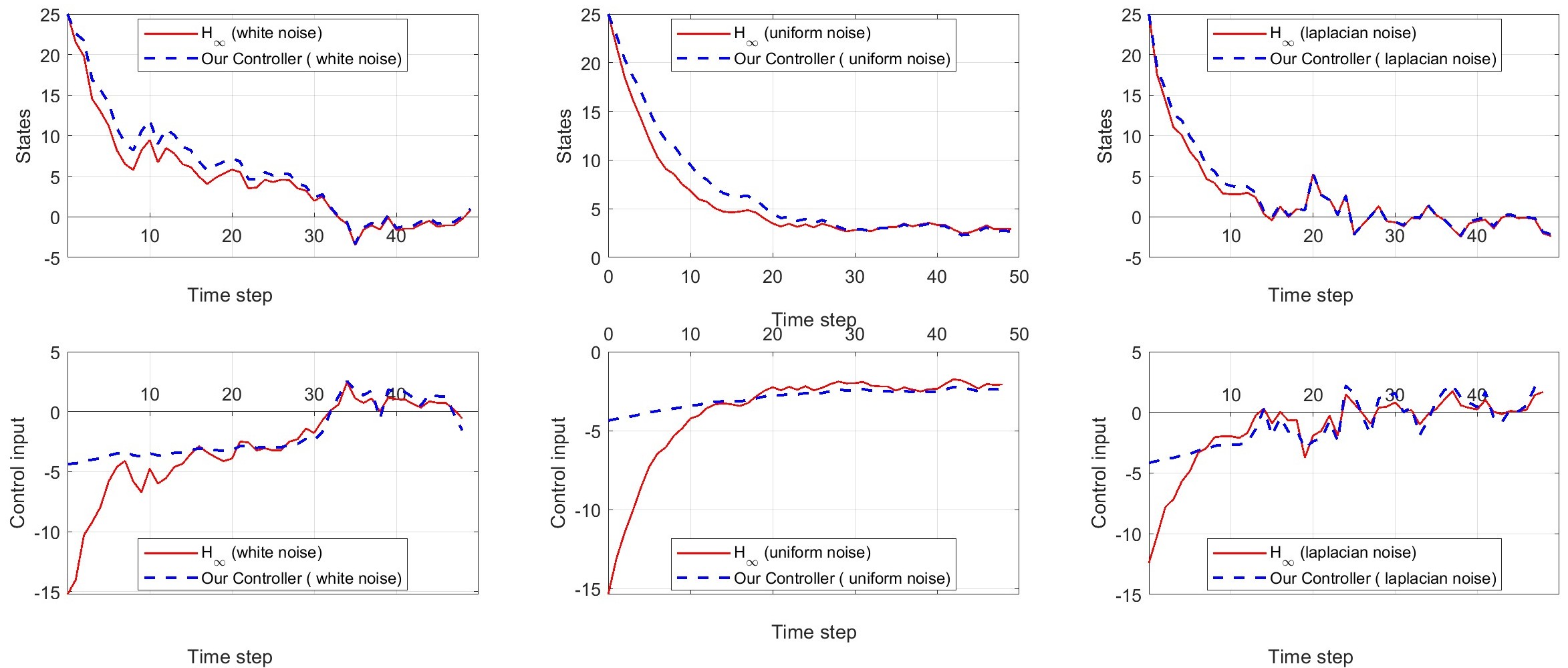}
  \caption{}
  \label{fig:rEXP}
\end{figure}

\section{Conclusion}
This work extends \(H_\infty\) control from quadratic costs to more general cost formulations, while retaining an explicit controller design. Future work will focus on constructing cost functions $(q,r,s)$ that satisfy the conditions of Theorem~\ref{thm:main} through exploring nonquadratic functions $(m,g)$.

\bibliographystyle{ieeetr}
\bibliography{refs}
\appendix
\subsection{Extended Proof of Theorem \ref{thm:main}}\label{app:B}

Our aim is to extend the argument in Section~\ref{sec:proofmain}—which currently assumes $B$ is invertible—to the general case where $n>d$ or $n<d$.  Indeed, the only step that invokes invertibility of $B$ is the passage from \eqref{eq:bregdivcost} to \eqref{eq:finaljt}.  Hence, we will show that \eqref{eq:finaljt} is equal to \eqref{eq:bregdivcost} without inverting $B$.

Given equation \eqref{eq:bregdivcost}, we have that \eqref{eq:st1}, \eqref{eq:st2} and \eqref{eq:riccati-like} hold.
Then, we define $g_f(.)=g^\ast(.)=r^\ast (B^\top(.))+p^\ast (.)$. Thus, Lemma \ref{lem:aux} and equations \eqref{eq:grp}, \eqref{eq:sudden} hold.

Now, we show that
\[
\Delta \;:=\; \bigl(\eqref{eq:bregdivcost}\bigr)
\;-\;
\bigl(\eqref{eq:finaljt}\bigr)
\]
vanishes identically.
Dropping the time index for clarity, we have:
$\Delta=D_{\gamma^2s}(w,\hat w)-D_r(u,\hat u)-D_p(Ax+Bu+w,Ax+B\hat u+\hat w)-D_{\gamma^2s}(w,\hat w)+D_g(Ax+w,Ax+\hat w)+D_r(u,u^\ast )+D_p(Ax+Bu+w,Ax+B u^\ast +w)$. Expanding the Bregman divergence terms through their definitions, we get: $\Delta
= \sum_{j=1}^{16} T_j,$
where each \(T_j\) is marked by an underbrace (i)–(xvi).  Concretely,

\begin{align*}
&\Delta= \underbrace{-r(u)}_{i}
  + \underbrace{r(\hat u)}_{ii}
  + \underbrace{(u-\hat u)^\top\nabla r(\hat u)}_{iii}
   \underbrace{-p(Ax+Bu+w)}_{iv}\\
  &+ \underbrace{p(Ax+B\hat u+\hat w)}_{v}
  + \underbrace{(u-\hat u)^\top B^\top\nabla p(Ax+B\hat u+\hat w)}_{vi}\\
  &+ \underbrace{(w-\hat w)^\top\nabla p(Ax+B\hat u+\hat w)}_{vii}
  + \underbrace{g(Ax+w)}_{viii}\\
  & \underbrace{-g(Ax+\hat w)}_{ix}
   \underbrace{-(w-\hat w)^\top\nabla g(Ax+\hat w)}_{x}
\\
  &+ \underbrace{r(u)}_{xi}
   \underbrace{-r(u^\ast)}_{xii}
  + \underbrace{(u-u^\ast)^\top\nabla r(u^\ast)}_{xiii}
  + \underbrace{p(Ax+Bu+w)}_{xiv}
\\
  & \underbrace{-p(Ax+B u^\ast + w)}_{xv}
   \underbrace{-(u-u^\ast)^\top B^\top\nabla p(Ax+B u^\ast + w)}_{xvi}
\end{align*}
Note that: a) $(i)+(xi)=0$, b) $(ii)+(v)+(ix)=0$ using \eqref{eq:sudden}, c) $(iii)+(vi)=0$ using \eqref{eq:st1}, d)  $(iv)+ (xiv)=0$ , e) $(vii)+(x)=0$ using Lemma \ref{lem:aux}, f) $(viii)+(xii)+(xv)=0$ using \eqref{eq:grp}, g) $(xiii)+(xvi)=0$ using Lemma \ref{lem:aux}. Thus all terms cancel yielding \(\Delta=0\) and hence
$\eqref{eq:bregdivcost}=\eqref{eq:finaljt}$ which concludes the proof.

\subsection{Proof of Corollary \ref{corr:1}}\label{app:corr}

Assume quadratic costs
\(
q(x)=x^\top Qx,\;
r(u)=u^\top Ru,\;
s(w)=w^\top Sw
\)
with \(Q,R,S\succ0\), and let \(p(x)=x^\top Px\) with \(P\succ0\).
Thus we have:
$
p^\ast(\xi)=\frac{1}{4}\,\xi^\top P^{-1}\xi,\quad
r^\ast(B^\top\xi)=\frac{1}{4}\,\xi^\top BR^{-1}B^\top\xi,$ $
(p-q)^\ast(A^\top\xi)=\frac{1}{4}\,\xi^\top A(P-Q)^{-1}A^\top\xi,
$
and
\(
s^\ast(\gamma^{-2}\xi)=\frac{1}{4}\,\gamma^{-4}\,\xi^\top S^{-1}\xi.
\)

Plugging these into the Riccati-like identity \eqref{eq:RICC2} gives, for all \(\xi\),
\[
\frac{1}{4}\,\xi^\top\!\bigl(P^{-1}\!+\!BR^{-1}B^\top\bigr)\xi
\;=\;
\frac{1}{4}\,\xi^\top\!\bigl(A(P\!-\!Q)^{-1}\!A^\top+\gamma^{-2}S^{-1}\bigr)\xi,
\]
hence the matrix equality
\begin{equation}\label{eq:invform}
P^{-1}+BR^{-1}B^\top
\;=\; A(P-Q)^{-1}A^\top+\gamma^{-2}S^{-1},
\end{equation}
with $G:=\bigl(P^{-1}+BR^{-1}B^\top\bigr)^{-1}.$
Equivalently,
\begin{equation}\label{eq:invRicc}
    P
=Q+A^\top\!\bigl(P^{-1}+BR^{-1}B^\top-\gamma^{-2}S^{-1}\bigr)^{-1}\!A.
\end{equation}
This is an inverse-form discrete-time \(H_\infty\) ARE.
\paragraph*{(i) Recovering the classical Riccati equation}

Let
$$H\;:=\;P^{-1}+BR^{-1}B^\top-\gamma^{-2}S^{-1},$$ 
and introduce the block matrix
\(
\mathcal{M}=\begin{bmatrix}
\operatorname{diag}(-R,\,\gamma^2 S) & \begin{bmatrix} B^\top \\[1pt] I \end{bmatrix}\\[4pt]
\begin{bmatrix} B & I \end{bmatrix} & P^{-1}
\end{bmatrix}
\);
its Schur complements satisfy \(\mathcal{M}/\operatorname{diag}(-R,\gamma^2 S)=H\) and \(\mathcal{M}/P^{-1}=-R_e^c\), where
$
R_e^c \;=\;
\begin{bmatrix}
R + B^\top P B & B^\top P \\
P B & -\gamma^2 S + P
\end{bmatrix}
$. Then, by the block inverse formula, we get:
\begin{equation}\label{eq:keyHinv}
\;
(\mathcal{M}^{-1})_{22}=H^{-1}
\;=\;
P \;-\; P\,[\,B\;\; I\,]\,(R_e^c)^{-1}\!\begin{bmatrix} B^\top \\[2pt] I \end{bmatrix} P.
\;
\end{equation}

Plugging \eqref{eq:keyHinv} in \eqref{eq:invRicc} gives
\begin{align}
P
&= Q + A^\top P A
- A^\top P [\,B\;\; I\,](R_e^c)^{-1}\!\begin{bmatrix} B^\top \\[1pt] I \end{bmatrix} P A \nonumber\\
&= A^\top P A + Q - K_c^\top R_e^c K_c,
\end{align}
with
\(K_c:=(R_e^c)^{-1}\!\begin{bmatrix} B^\top \\[1pt] I \end{bmatrix} P A\).
This is the \emph{steady-state} (stationary infinite-horizon) quadratic \(H_\infty\) Riccati equation of Theorem~\ref{thm:robust_ctrl}.

\paragraph*{(ii) Recovering the negativity condition}
With
\(
G=\bigl(P^{-1}+BR^{-1}B^\top\bigr)^{-1},
\)
the matrix inversion lemma gives,
\begin{equation}\label{eq:Gwoodbury}
G
\;=\;
P - P B (R + B^\top P B)^{-1} B^\top P.
\end{equation}
Condition \eqref{eq:concave} becomes $-\gamma^2S+G\preceq0$.
Replacing $G$ with \eqref{eq:Gwoodbury}, we get:
$$-\gamma^2S+P - P B (R + B^\top P B)^{-1} B^\top P\preceq 0,$$
which is the the standard negativity condition in steady state. 

\paragraph*{(iii) Recovering the linear controller}
With quadratic \(r,g\): \(\nabla r^\ast(y)=\frac{1}{2} R^{-1}y\), \(\nabla g(z)=2Gz\).
Thus the optimal controller \eqref{eq:optimal_controller} becomes
$u^\ast(x,w)
= -\,\nabla r^\ast\bigl(B^\top \nabla g(Ax+w)\bigr)
= -\,R^{-1}B^\top G(Ax+w).
$
Using \(R^{-1}B^\top(P^{-1}+BR^{-1}B^\top)^{-1}=(R+B^\top PB)^{-1}B^\top P\) and \eqref{eq:Gwoodbury} gives
\[
u^\ast(x,w)=-(R+B^\top P B)^{-1} B^\top P\,(Ax+w),
\]
i.e., the classical full-information law in steady state.

Thus, under quadratic costs, the Riccati-like identity \eqref{eq:RICC2}, the concavity test \eqref{eq:concave}, and the controller \eqref{eq:optimal_controller} reduce exactly to the classical discrete-time \(H_\infty\) ARE, negativity condition, and linear central controller, respectively.

\subsection{Proof Idea of Theorems \ref{thm:chooseR}, \ref{thm:chooseQ} and \ref{thm:chooseQR_} }\label{app:1}

The proofs rely on the following ideas:

\begin{enumerate}
    \item \textbf{Convexity via Hessian Positivity:}  
    A function \( f \) is convex if
    \[
    \nabla^2 f(x) \succ 0 \quad \forall x.
    \]
    \item By imposing Assumption \ref{asum:q} as a constraint, convexity implies positivity of the functions.
    
    \item \textbf{Fenchel Duality and the Inverse Function Theorem:}  
    For a convex function \( f \) and its Fenchel conjugate \( f^\ast \), one obtains (via the inverse function theorem)
    \[
    \nabla^2 f\bigl(\nabla f^\ast(\xi)\bigr)\,\nabla^2 f^\ast(\xi)=I.
    \]This crucial relation shows that locally the Hessian of \( f\) is the inverse of the Hessian of \( f^\ast \).

\end{enumerate}
Combining these ideas along with some matrix inequalities provides sufficient conditions for the overall functions to be convex.

\subsection{Proof of Theorem \ref{thm:chooseR}}\label{app:2}
% First, equate the primal and dual expressions:
In order for equation \eqref{eq:RICC2} to hold, we need to ensure $\gamma^2s(\cdot)-g(Ax+(\cdot))$ is convex for a fixed $x$, and $p(\cdot)$ and $q(.)$ are convex (which will imply $p(\cdot)$ and $q(.)$ are positive as will be seen in the proof).

To ensure the convexity of $p(.)$, we utilize the following result from convex analysis:
\begin{lemmma}\label{lemm:posFenchel}
    If a function $f$ is convex, positive, and $f(0)=0$, then its Fenchel dual $f^\ast$ is also convex, positive and $f^\ast(0)=0$, and vice versa. 
\end{lemmma}
Then, using Lemma \ref{lemm:posFenchel}, we will observe that $p^{\ast\ast}$ is convex, even, positive and by choosing $p:= p^{\ast\ast}$, the results follow. 
Thus, for the quadratic function $m(x)=x^\top  M x$, convexity of $p$ is equivalent to
  \begin{equation}\label{eq:conv_p_hessian}
    \nabla^2 p^\ast \succeq 0,
  \end{equation}
  and from \ref{eq:RICC2} we have:
  \begin{align}\label{eq:key_operator}
    &\nabla^2 p^\ast(\xi)=A\nabla^2 m^\ast(A^\top  \xi)A^\top 
    +\gamma^2\nabla^2 s^\ast(\gamma^{-2}\xi)\nonumber\\
    &\quad \quad \quad \quad -B\nabla^2 r^\ast(B^\top  \xi)B^\top \\
    &\text{thus we need: }\\
    &B\nabla^2 r^\ast(B^\top  \xi)B^\top 
    - \gamma^2\nabla^2 s^\ast(\gamma^{-2}\xi)
    \preceq A\nabla^2 m^\ast(A^\top  \xi)A^\top .
  \end{align}

  % Overleaf snippet

Now, after ensuring convexity of $p^\ast$, note that its gradient at the origin is given by
$\nabla p^\ast(0)
= A\,\nabla m^\ast(A^\top 0)\,A^\top
+ \nabla s^\ast(\gamma^{-2}0)
- B\,\nabla r^\ast(B^\top 0).
$
By Assumption \ref{asum:q}, $\nabla s^\ast(0)=\nabla r^\ast(0)=0$, and since $\nabla m^\ast(x)=\frac{1}{2}M^{-1}x$ we have $\nabla m^\ast(0)=0$.  Hence $\nabla p^\ast (0)=0$, which for a convex function implies that the origin is its global minimizer and therefore $p^\ast (x)\ge0$ for all $x$.  Similarly, because $\nabla q(0)=\nabla p(0)-\nabla m(0)=0$, convexity of $q$ also guarantees $q(x)\ge0$ for all $x$.  We thus turn to establishing the convexity of $q$.

  To ensure the convexity of $q$, we require
  \begin{equation}\label{eq:conv_q_hessian}
    \nabla^2 p(x) - 2M \succeq 0,
    \quad\forall x
  \end{equation}
  Using the inverse function theorem (see Section~\ref{app:1}), we then have
  \begin{equation}\label{eq:conj_hess_bound}
    \nabla^2 p^\ast (\nabla p(x)) \preceq \frac{1}{2} M^{-1},
    \quad\forall x,
  \end{equation}
 Taking $\tilde{x}:=\nabla p(x)$, we observe that since \eqref{eq:conj_hess_bound} holds for every $x$
  and $\nabla p$ is bijective,
  \begin{equation}\label{eq:conj_hess_global}
    \nabla^2 p^\ast (\tilde x) \preceq \frac{1}{2} M^{-1},
    \quad\forall\,\tilde x.
  \end{equation}
which gives 
  \begin{align}
    &\nabla^2 p^\ast \!(\cdot)\!=\!\frac{1}{2} A M^{-1}\!\!A^\top \!\! +\! \gamma^2\nabla^2 s^\ast \!(\gamma^{-2}(\cdot))
    \!-\! B\nabla^2 r^\ast \!(B^\top \!(\cdot))B^\top\! \nonumber\\
    &\preceq \frac{1}{2} M^{-1}
  \end{align}
  %Combining \eqref{eq:key_operator} and \eqref{eq:conj_hess_global} gives the bound
or equivalently,
  \begin{align}
    &B\nabla^2 r^\ast (B^\top (\cdot))B^\top 
    \;-\;\gamma^2\nabla^2 s^\ast (\gamma^{-2}(\cdot))\nonumber\\
    &
    \succeq \frac{1}{2} A M^{-1}A^\top  - \frac{1}{2} M^{-1}.
  \end{align}
  Therefore, combining \eqref{eq:key_operator} and \eqref{eq:conj_hess_global} gives the two-sided bound which ensures convexity of both $p$ and $q$:
  \begin{align}\label{eq:pqConvex_restate}
    L \preceq B\nabla^2 r^\ast (B^\top  \xi)B^\top 
    - \gamma^{-2}\nabla^2 s^\ast (\gamma^{-2}\xi) \preceq U,
    \quad\forall\xi,
  \end{align}
  where
  \[ L:=\frac{1}{2}\bigl(A M^{-1}A^\top  - M^{-1}\bigr),
     \qquad U:=\frac{1}{2} A M^{-1}A^\top . \]
  Finally, note that if
  \[ \inf_\xi \bigl[B\nabla^2 r^\ast (B^\top \xi)B^\top  - \gamma^{-2}\nabla^2 s^\ast (\gamma^{-2}\xi)\bigr] = 0, \]
  then \eqref{eq:pqConvex_restate} forces
  \[ \frac{1}{2}\bigl(A M^{-1}A^\top  - M^{-1}\bigr) \preceq 0, \]
  implying stability of $A$ since $M\succ 0$. This justifies the stability assumption in the first statement of the theorem.

  Ensuring the convexity of $\gamma^2s(\cdot)-g(Ax+(\cdot))$ via \eqref{eq:MMM} in Theorem \ref{thm:chooseR} follows a similar line of arguments (convex analysis and Inverse function theorem) and will be ommitted.

\paragraph{Strong Convexity and Smoothness}
Under the extra assumption that $r^\ast $ and $s^\ast $ are both strongly convex and smooth, the second part of the theorem follows easily. Recall the elementary fact: if two matrix‐valued functions satisfy
$
X_{\min}\;\preceq\;X(t)\;\preceq\;X_{\max}
$$
\text{and } Y_{\min}\;\preceq\;Y(t)\;\preceq\;Y_{\max}
\quad\forall\,t,
$
then a sufficient condition for \(X(t)\preceq Y(t)\) for all \(t\) is
$
\sup_t X(t)\;\preceq\;\inf_t Y(t),
\quad\text{i.e.}\quad
X_{\max}\;\preceq\;Y_{\min}.
$
We apply this observation to obtain every inequality in \eqref{eq:cvxopt} by bounding each matrix function above and below by its extremal values. The only exception is
$B\,U_r^{-1}B^\top \;-\;\gamma^2\,L_s^{-1}\;\succ\;0$
which must be imposed directly when \(A\) may be unstable. 

This concludes the proof of Theorem \ref{thm:chooseR}.

\end{document}